\documentclass[reqno,a4paper,11pt]{amsart}
\usepackage[english]{babel}
\usepackage{color}
\usepackage{soul}
\usepackage{enumerate}

\usepackage{amssymb,accents}
\usepackage{mathrsfs, mathtools}

\usepackage[all]{xy} 
\input xy 
\xyoption{all}
\xyoption{2cell} 
\xyoption{v2}
\SelectTips{cm}{11} 

\usepackage{subdepth}

\usepackage{hyperref}

\DeclareMathOperator{\sign}{sign}

\DeclareMathOperator{\vol}{vol}

\DeclareMathOperator{\id}{id}
\DeclareMathOperator{\hol}{hol}

\DeclareMathOperator{\DD}{DD}

\DeclareMathOperator{\pt}{pt}

\DeclareMathOperator{\im}{im}
\DeclareMathOperator{\ii}{i}

\theoremstyle{plain}
\newtheorem{theorem}[equation]{Theorem}
\newtheorem{corollary}[equation]{Corollary}

\newtheorem{proposition}[equation]{Proposition}

\theoremstyle{definition}
\newtheorem{definition}[equation]{Definition}

\theoremstyle{remark}
\newtheorem{remark}[equation]{Remark}
\newtheorem{example}[equation]{Example}

\numberwithin{equation}{section}
\numberwithin{figure}{section}

\renewcommand{\cH}{{\mathcal H}}
\newcommand{\cG}{\mathcal{G}}

\newcommand{\cU}{{\mathcal U}}

\newcommand{\cP}{{\mathcal P}}

\newcommand{\To}{{\ \Rightarrow\ }}
\newcommand{\dd}{{\rm d}}

\newcommand{\CC}{{\mathbb C}}
\newcommand{\RR}{{\mathbb R}}
\newcommand{\ZZ}{{\mathbb Z}}

\newcommand{\HH}{{\mathbb H}}

\renewcommand{\a}{\alpha}

\renewcommand{\d}{\delta}

\newcommand{\KR}{K\hspace{-1pt}R}

\newcommand{\KO}{K\hspace{-1pt}O}

\newcommand{\KSp}{K\hspace{-1pt}Sp}

\newcommand{\<}{\langle}
\renewcommand{\>}{\rangle}

\makeatletter
\newlength{\@thlabel@width}%
\newcommand{\thmenumhspace}{\settowidth{\@thlabel@width}{(1)}\sbox{\@labels}{\unhbox\@labels\hspace{\dimexpr-\leftmargin+\labelsep+\@thlabel@width-\itemindent}}}
\makeatother

\begin{document}

\begin{flushright}

\end{flushright}

\vskip 1cm

\title[Sign choices for orientifolds]{Sign choices for orientifolds}

  \author[P. Hekmati]{Pedram Hekmati}
  \address[Pedram Hekmati]
  {Department of Mathematics, University of Auckland, 1010 Auckland, New Zealand}
  \email{p.hekmati@auckland.ac.nz}

  \author[M.K. Murray]{Michael K.~Murray}
  \address[Michael K.~Murray]
  {School of Mathematical Sciences\\
  University of Adelaide\\
  Adelaide, SA 5005 \\
  Australia}
  \email{michael.murray@adelaide.edu.au}

  \author[R.J. Szabo]{Richard J.~Szabo}
  \address[Richard J.~Szabo]
  {Department of Mathematics, Maxwell Institute for Mathematical Sciences and The Higgs Centre for Theoretical Physics\\
  Heriot-Watt University\\
  Edinburgh, EH14 4AS \\
  United Kingdom}
  \address{Dipartimento di Scienze e Innovazione
  Tecnologica\\ and INFN Torino, Gruppo collegato di
  Alessandria, Universit\`a del Piemonte Orientale, Alessandria, Italy\\ and
Arnold--Regge Centre, Torino, Italy
  }
  \email{r.j.szabo@hw.ac.uk}
  
  \author[R.F. Vozzo]{Raymond F.~Vozzo}
\address[Raymond F.~Vozzo]
{School of Mathematical Sciences\\
University of Adelaide\\
  Adelaide, SA 5005 \\
  Australia}
\email{raymond.vozzo@adelaide.edu.au} 

\thanks{The authors acknowledge support under the  Australian
Research Council's {\sl Discovery Projects} funding scheme (project numbers DP120100106, DP130102578 and DP180100383), the Consolidated Grant ST/P000363/1 
from the UK Science and Technology Facilities Council, the Marsden Foundation (project number 3713803), and the Action MP1405 QSPACE from the European Cooperation in Science and Technology
(COST).\\
We thank David Baraglia and Jonathan Rosenberg for helpful discussions.}

\subjclass[2010]{}

\begin{abstract}  We analyse the problem of assigning sign choices to O-planes in orientifolds of type~II string theory. We show that there exists a sequence of invariant $p$-gerbes with $p\geq-1$, which give rise to sign choices and are related by coboundary maps. We prove that the sign choice homomorphisms stabilise with the dimension of the orientifold and  we derive topological constraints on the possible sign configurations. Concrete calculations for spherical and toroidal orientifolds are carried out, and in particular  we exhibit a four-dimensional orientifold where not every  sign choice is geometrically attainable. We elucidate how the $K$-theory groups associated with invariant $p$-gerbes for $p=-1,0,1$ interact with the coboundary maps. This allows us to interpret a notion of $K$-theory due to Gao and Hori as a special case of twisted $\KR$-theory, which consequently implies the homotopy invariance and Fredholm module description of their construction.  \end{abstract}
\maketitle
\tableofcontents

\bigskip

\section{Introduction\label{sec:Intro}} 

In this paper we study the Real Brauer group and related structures on {orientifolds}, that
is, pairs $(M,\tau)$ consisting of a manifold $M$ equipped with an
involution $\tau\colon M\to M$. The fixed point set of the involution is $M^\tau$ and its connected components will be
called {orientifold planes} or {O-planes} for short. Orientifolds are more commonly refered to as
Real manifolds in the mathematics literature; the present terminology is borrowed
from string theory where these give backgrounds
which are important for model building and understanding
T-duality~\cite{AngSag,DFM}. 

O-planes can carry a positive or negative
Ramond--Ramond charge, which
for consistency of the string background must be cancelled by the
inclusion of suitable D-brane
charges; conversely, string backgrounds with D-branes require
inclusion of suitable orientifold planes. The collection of discrete O-plane charges for a given orientifold
background is specified mathematically by a sign choice, which is an assignment of an element $\pm\,1$ of $\ZZ_2$ to each O-plane,
or equivalently a cocycle in $H^0(M^\tau, \ZZ_2)$. The sign choices
determine the projection of the $U(n)$ gauge groups of the Chan--Paton factors when
$n$ coincident D-branes are placed on top of an O-plane: this is $O(n)$ for sign
choice $+1$ and $Sp(n)$ for sign choice $-1$. 

A general condition for
the allowed distributions of sign choices for a given orientifold
$(M,\tau)$ is not presently known. A related open problem is to define a
suitable variant of $\KR$-theory classifying D-branes on orientifolds
for an arbitrarily given sign choice which is not necessarily
constant over the different O-planes. The $K$-theory and cohomology classification of the spectrum
of O-plane charges for a given orientifold is discussed in~\cite{Witten,Gukov,Hori,OlsSza,deBDHKMMS,BGS}. Real continuous trace
$C^*$-algebras and their sign choices were discussed in
\cite{Ros2013}, where a corresponding version of $\KR$-theory was also proposed. The purpose of this paper is to
provide some detailed answers to these questions from a geometric
perspective.

In orientifolds of type~II string theory, the behaviour of fields in
the NS--NS sector influences the allowed sign choices. At
the lowest level there is the dilaton field, which is a function on
spacetime $M$. Next, a globally defined Neveu--Schwarz $B$-field plays the role of a magnetic
field on a D-brane and so is naturally associated to a line bundle on
spacetime. More generally, a $B$-field with non-vanishing $H$-flux is
associated with a bundle gerbe on spacetime; a non-trivial $2$-torsion
$H$-flux on an O-plane reverses its sign choice, as will also follow
by our geometric considerations below. 

We therefore wish to consider the mathematical problem of assigning a sign choice to certain geometric objects over $(M,\tau)$. 
For instance in the simplest case of a Real function, that is, a
function $f \colon M \to
U(1)$ which satisfies $f \circ \tau = \bar f$, we have $f(m) \in
\ZZ_2$ at a fixed point $m\in M^\tau$ so  
it defines a sign choice. More generally, we show that there is a
sequence of natural geometric objects on $M$ to which sign choices can be allocated. We also address the fundamental question of whether every 
sign choice can arise from one of these geometric objects,
particularly those that encode the data of a string compactification.  

Of particular  interest, from the string theory perspective, will be the first three classes of these geometric 
objects: Real  functions as discussed above, $\ZZ_2$-equivariant line
bundles, and Real bundle gerbes. These have  characteristic classes in
the equivariant sheaf cohomology groups $H^0(M; \ZZ_2, \cU^1)$,
$H^1(M; \ZZ_2, \cU^0)$ and $H^2(M; \ZZ_2, \cU^1)$, respectively, that
we will discuss below.  Each of these groups comes equipped with a {\em sign choice map} which is a homomorphism  $\sigma_i$ mapping into $H^0(M^\tau, \ZZ_2)$. There are connecting homomorphisms $\partial_i$ between these cohomology groups giving rise to the commutative diagram
\begin{equation}\label{eq:connhom}
 \begin{gathered}   \xymatrix{
         H^0(M; \ZZ_2, \cU^1)     \ar@{->}[dr]^{ \sigma_0} \ar@{->}[d]^{ \partial_1}  &      \\
                H^1(M; \ZZ_2, \cU^0)        \ar@{->}[r]^{ \sigma_1}    \ar@{->}[d]^{ \partial_2}      &\ \ \ \  H^0(M^\tau, \ZZ_2)    \\
        H^2(M; \ZZ_2, \cU^1) \ar@{->}[ur]_{\sigma_2}\ar@{->}[d]^{ \partial_3}  &  \\
           \vdots  &  }
        \end{gathered}
\end{equation}
As the vertical arrow $\partial_3$ suggests, this is part of a more general story. 
We denote by ${\mathit\Sigma}_i(M)$ the image of the sign choice homomorphisms $\sigma_i$ in $H^0(M^\tau, \ZZ_2)$ and  call such elements {\em geometric sign choices}. 
  If $M$ has dimension $d$, we find that these groups coincide in
  degrees $d-1$ onwards. As a result, 
$$
{\mathit\Sigma}_0(M)  \subseteq {\mathit\Sigma}_1(M) \subseteq \cdots \subseteq
{\mathit\Sigma}_{d-1}(M) = {\mathit\Sigma}_{d}(M) = \cdots  \subseteq H^0(M^\tau, \ZZ_2)
\ ,
$$
and thus  the geometric sign choices are precisely ${\mathit\Sigma}_{d-1}(M)$.
We show that not all sign choices are geometric sign choices by 
exhibiting a four-manifold $M$ where ${\mathit\Sigma}_3(M) \neq H^0(M^\tau, \ZZ_2)$.  

In Section \ref{sec:constraining} we derive some constraints on the
possible geometric sign choices, that is, on the images ${\mathit\Sigma}_i(M)$
for $i=0,1,2$. For a Real bundle gerbe $(P,Y)$ over a $2$-connected
orientifold $M$ with vanishing complex Dixmier--Douady class $\DD(P)$, we prove that the 
sign choice is generically constant. More generally, we show in Proposition
\ref{prop:constraint} that the sign choice map $\sigma_2$ is 
completely determined by the periods of the
complex Dixmier--Douady class and the sign at a single fixed point  by the formula
 $$
 \sigma_2(P)(m) = (-1)^{\< \mu_{m,m'}^*(\DD(P)), S^3\>} \, \sigma_2(P)(m'\,) \ ,
 $$
 where the map $\mu_{m,m'} \colon S^3 \to M$ is constructed from the
 $\tau$-action and the fixed points $m,m'\in M^\tau$.  This kind of construction 
 depends only on knowing the sign choices for spheres and the formula extends to all invariant $p$-gerbes under appropriate connectedness conditions on $M$.
 In particular for the sign choice homomorphisms $\sigma_0$
 and $\sigma_1$ on connected respectively $1$-connected orientifolds,
 the Dixmier--Douady class is substituted by the winding class
 respectively first Chern class.
  
 We proceed to consider some concrete examples of orientifolds which
 appear in conventional string theory backgrounds. An important class
 of compactification spaces arising in string theory with background fluxes
 are spheres. Recall that $S^{p, q} \subset \RR^{p+q}$ 
denotes the $p+q-1$-dimensional sphere with the involution $\tau$
which changes the sign of the first $p$ coordinates
and leaves alone the remaining $q$ coordinates. We also carry out explicit calculations involving tori,
which play an important role in flat compactifications of string
theory and provide examples of multiply connected orientifolds.
The results are summarised in the following table:\\

\renewcommand{\arraystretch}{1.5}
\begin{center}
\begin{tabular}{|l|c|c|c|c|c|c| c|}
\hline
$(M,\tau)$  & $S^{3, 1}$& $S^{1, 1} $& $S^{1, 1}\times S^{1, 1}$  &$S^{1, 1}\times S^{1, 1}\times S^{1, 1}$   \\
\hline
\hline
$\dim(M)$ &  $3$ &  $1$ &  $2$ &  $3$  \\
\hline
$H^0(M^\tau, \ZZ_2) $  & $\ZZ_2^2$  &$\ZZ_2^2$ & $ \ZZ_2^4 $ & $\ZZ_2^8$ \\
\hline

${\mathit\Sigma}_0(M) $    & $\ZZ_2$ & $\ZZ_2^2$ & $\ZZ_2^3$ & $\ZZ_2^4$ \\
\hline
$H^1(M; \ZZ_2, \cU^0) $  & $\ZZ_2$   & $\ZZ_2^2$   & $\ZZ_2^3\oplus\ZZ$  & $\ZZ_2^4\oplus\ZZ^3$  \\

${\mathit\Sigma}_1(M) $    &   $\ZZ_2$  & $\ZZ_2^2$ & $ \ZZ_2^4 $& $\ZZ_2^7$ \\
\hline
$H^2(M; \ZZ_2, \cU^1)$  & $\ZZ_2 \oplus \ZZ$  &  $\ZZ_2^2$& $ \ZZ_2^3\oplus\ZZ $ &  $\ZZ_2^7 \oplus \ZZ$ \\

 ${\mathit\Sigma}_2(M)$&  $\ZZ_2^2$  & $\ZZ_2^2$ &$ \ZZ_2^4 $ & $\ZZ_2^8$\\

\hline
\end{tabular}\\
\end{center}

\bigskip

Note that ${\mathit\Sigma}_2(M) = H^0(M^\tau, \ZZ_2)$ in these special cases. In particular, the three-dimensional examples show that,
generally, for all sign choices to arise one must consider Real
bundle gerbes which have non-vanishing complex Dixmier--Douady class in the
Real Brauer group $H^2(M; \ZZ_2, \cU^1)$. We also
consider higher-dimensional spheres $S^{n,1}$ with $n\geq4$, where the
sign choice map $\sigma_2$ is not surjective. In the case of tori, our geometric
sign choice groups ${\mathit\Sigma}_1(M)$ reproduce the classification
of O-plane charges from~\cite{deBDHKMMS} using the Borel equivariant
cohomology group $H_{\ZZ_2}^2(M,\ZZ_2)$ to determine the allowed
globally defined $B$-field configurations on $M$.

The possible sign choices for orientifolds were also considered from a
geometric perspective by Gao
and Hori in~\cite{GaoHor} using constructions based directly on type~II
string theory in backgrounds with topologically trivial $H$-flux. 
In Section \ref{sec:KR-theory} we show how to interpret their
geometric data in terms of our constructions, and 
demonstrate that their version of $\KR$-theory can be understood as being twisted by either a Real function or an equivariant line bundle.  We show that in either
case it is equivalent to the $\KR$-theory twisted by a corresponding
Real bundle gerbe constructed by repeated application of the connecting 
homomorphisms
$$
H^{0}(M; \ZZ_2, \cU^1) \xrightarrow{ \ \partial_1 \ } H^{1}(M; \ZZ_2,
\cU^0) \xrightarrow{ \ \partial_2 \ } H^{2}(M; \ZZ_2, \cU^1) \ .
$$
In these cases the $\KR$-theory classes are twisted by a Real bundle
gerbe with trivial complex Dixmier--Douady class, whereby the possible
Real structures on the bundle gerbe are given by equivariant line
bundles (it is a torsor under $H^1(M;\ZZ_2,\cU^0)$). The spectrum
of O-plane charges is then determined by the fact that, on the connected
components of the fixed point set of the orientifold involution, the $\KR$-theory
becomes $\KO$-theory when the restriction of the bundle gerbe has
sign $+1$ and so is Real
trivial,
or $\KSp$-theory when the restriction of the bundle gerbe has
$2$-torsion Dixmier--Douady class with sign $-1$ corresponding to
non-trivial torsion $H$-flux~\cite{HekMurSza}. However, in general there are sign choices
that cannot be achieved in such a way unless the complex
Dixmier--Douady class (or $H$-flux of the string background) is non-trivial,\footnote{This seems to be at odds with the result of ~\cite{Ros2013} who used $C^*$-algebra
methods to show that any sign choice could arise, whereas we have found constraints.
This raises the possibility that the  correspondence
between  continuous trace $C^*$-algebras and groupoids for the Real case
is not  a  straightforward generalisation of the complex case.} in which case the more general geometric realisation of twisted $\KR$-theory from~\cite{HekMurSza} contains the appropriate sign choices.

A  summary of the contents of this paper is as follows. In Section \ref{sec:signs} we introduce invariant $p$-gerbes for $p\geq -1$ and define their associated sign choice homomorphisms at the level of cohomology. We derive some general properties of the sign choice maps and show that they stabilise at degree $\dim(M) -1$. In Section  \ref{sec:realgerbes} we restrict our attention to the case $p=1$, and explain how to incorporate the orientifold $B$-field and its holonomy.  As an application, we obtain sufficient conditions for the existence of Real structures on bundle gerbes when $M$ is 2-connected. We proceed to give geometric constructions of the sign choice maps in Section~\ref{sec:bgsign} for Real functions, equivariant line bundles and Real bundle gerbes. We further derive a general characterisation of the kernel and image of the sign
choice map $\sigma_2$ in terms of the sheaf cohomology of the space of
``physical points'' $M\slash \tau$ of the orientifold, which exhibits obstructions to the possible sign choices. In Sections~\ref{sec:constraints} and~\ref{prop:nonsimplyconn} we carry out detailed calculations for some spherical and toroidal orientifolds. Under certain connectedness conditions on $M$, we obtain a geometric constraint on the sign choices in terms of the periods of the complex characteristic classes of the $p$-gerbes. Finally, in Section~\ref{sec:KR-theory} we recast the orientifold data of Gao and Hori  \cite{GaoHor} in the language of the present paper and  prove that their notion of {\em $K$-theory twisted
  by an equivariant line bundle $L$} is nothing but  $\KR$-theory twisted by the Real bundle gerbe $\partial_2(L)$ under the connecting homomorphism in~\eqref{eq:connhom}.


\section{Sign choices for invariant $p$-gerbes}
\label{sec:signs}

In this section we develop some general theory of sign choices for {\em invariant $p$-gerbes}.
This is followed in subsequent sections by a discussion of the cases of particular geometric interest. 
Let $(M,\tau)$ be an orientifold and let $M^\tau\subseteq M$ denote the fixed point set of the involution.

\begin{definition} A {\em sign choice}  is an element of the group $H^0(M^\tau, \ZZ_2)$.
\end{definition}
In other words, a sign choice assigns to every O-plane, that is, every
connected component of $M^\tau$, a sign $\pm\,1$ in $\ZZ_2$. 

We are interested in various geometric objects that determine a sign choice on an orientifold. 
For instance, Real functions, equivariant line bundles and Real bundle gerbes
form the first three examples of {invariant $p$-gerbes} corresponding to the cases $p = -1, 0, 1$.
The case $p = 2$ would correspond to equivariant bundle $2$-gerbes
which we briefly consider in Section \ref{2-gerbe}.  For $p > 2$
the geometric theory of $p$-gerbes is not well understood, but we can work 
instead with the corresponding cohomology groups that classify them since
the sign choice maps factorise through these groups.   

There are two natural $\ZZ_2$-sheaves on the orientifold which in \cite{HekMurSza}
we denoted by $\cU(1)$ and $\overline{\cU}{(1)}$.  They are the equivariant sheaves
of smooth functions into $U(1)$, where the first is endowed with a trivial 
$\ZZ_2$-action on $U(1)$ and the second has the $\ZZ_2$-action induced by complex
conjugation.  In the present paper it will be useful 
to change notation and denote these by $\cU^{\epsilon}$ with $\epsilon
\in \ZZ_2 = \{0, 1\}$, where
$\cU^0 = \cU(1)$ and $\cU^1 = \overline{\cU}(1)$. We will continue to
use the symbol $\cU(1)$ for the  sheaf of smooth $U(1)$-valued  functions.
With this notation it follows that $p$-gerbes, up to appropriate equivalence, are classified 
by $H^{p+1}(M, \cU(1))$.  We denote by $H^k(M; \ZZ_2, \cU^\epsilon)$  the Grothendieck equivariant sheaf cohomology \cite{Gro} of the 
equivariant sheaf $\cU^\epsilon$. 

An  {\em invariant $p$-gerbe} is a $p$-gerbe 
on an orientifold that is Real if $p$ is even and $\ZZ_2$-equivariant if $p$
is odd. They are classified up to appropriate equivalence by $H^{p+1}(M; \ZZ_2, \cU^{[p+2]})$, where
 $[k]$ denotes the reduction modulo $2$ of $k \in \ZZ$. 
Thus the first three types of invariant $p$-gerbes are:
\begin{itemize}
\item $p=-1$: Real functions from $M$ to $U(1)$ classified by $H^0(M; \ZZ_2, \cU^1)$.
\item $p=0$: equivariant line bundles on $M$ classified by $H^1(M; \ZZ_2, \cU^0)$.
\item $p=1$: Real bundle gerbes on $M$ classifed by $H^2(M; \ZZ_2, \cU^1)$.
\end{itemize}
Physically, for $p>1$ the invariant $p$-gerbes correspond respectively to backgrounds with Ramond--Ramond
fields on O-planes of dimension $p'\equiv p-2$ and $p'\equiv p-1$
modulo~$4$~\cite{BGS,Keur}. 

In particular when $M$ is a point
we have \cite{HekMurSza}
$$
H^{p+1}(\pt; \ZZ_2, \cU^{[p+2]}) = \ZZ_2
$$
and 
$$
H^{p+1}(\pt; \ZZ_2, \cU^{[p+1]}) = 0 \ .
$$
In other words, over a point there are exactly two equivalence classes of invariant $p$-gerbes. 
These are equalities as the group $\ZZ_2$ has no non-trivial automorphisms. 

If $m \in M^\tau$ is a fixed point, then the 
inclusion $\iota_m \colon \{m \} \hookrightarrow M$ is a morphism of
Real manifolds and we can use the corresponding pullbacks to define
restriction maps. 
If  $\xi \in  H^{p+1}(M; \ZZ_2, \cU^{[p+2]})$ we define 
$$
\sigma_{p+1}(\xi)(m) = \iota_m^*(\xi) \ \in \ H^{p+1}(\{m\}; \ZZ_2,
\cU^{[p+2]}) = \ZZ_2 \ .
$$
 Consider a continuous path $C$ in $M^\tau$ joining $m$ to $m'$.  Then if $p$ is odd we  have
$$
H^{p+1}(\{m\}; \ZZ_2,  \cU^{[p+2]}) = H^{p+1}(C; \ZZ_2,  \cU^{[p+2]})
= H^{p+1}(\{m'\}; \ZZ_2,  \cU^{[p+2]}) \ ,
$$
which shows that $\sigma_{p+1}(\xi)(m) = \sigma_{p+1}(\xi)(m'\,)$ and thus $\sigma_{p+1}(\xi) \in H^0(M^\tau, \ZZ_2)$.  We call this 
the sign choice of the $p$-gerbe $\xi$. 

From this construction it follows that the sign choice map 
\begin{equation}
\label{eq:signchoice}
\sigma_{p+1}:H^{p+1}(M; \ZZ_2, \cU^{[p+2]})\longrightarrow H^0(M^\tau,
\ZZ_2)
\end{equation}
is a homomorphism that commutes with pullbacks.  In fact,
it follows easily that any sign choice 
map that is a homomorphism and is compatible with pullbacks must be this one.  This will be useful 
when considering more geometric constructions of the sign choice map below. 

We showed in \cite{HekMurSza} that there exist long exact sequences relating the equivariant 
sheaf cohomology of $\cU^0$ and $\cU^1$.  For  $\epsilon  \in \ZZ_2$ we have
\begin{equation}
\label{eq:beginlongsequence}
\xymatrix@R=2ex{0 \ 
\ar[r] & \   H^0(M; \ZZ_2, \cU^\epsilon ) \ \ar[r] & \   H^0(M, \cU(1)) \  \ar[r]^{\!\!\!\!\!\!1 \times \tau^*} & \   H^0(M; \ZZ_2, \cU^{[\epsilon+1]} ) \ \ar`r[d]`[lll]`[ddlll] `[ddll][ddll]&\\
&&&&\\
 & \   H^1(M; \ZZ_2, \cU^\epsilon ) \ \ar[r] & \   H^1(M, \cU(1)) \ \ar[r]^{\!\!\!\!\!\!1 \times \tau^*} & \   H^1(M; \ZZ_2, \cU^{[\epsilon+1]} ) \ \ar`r[d]`[lll]`[ddlll] `[ddll][ddll]&\\
&&&&\\
 & \   H^2(M; \ZZ_2, \cU^\epsilon) \ \ar[r] & \   H^2(M, \cU(1)) \ \ar[r]^{\!\!\!\!\!\!1 \times \tau^*} & \   H^2(M; \ZZ_2, \cU^{[\epsilon+1]} ) \ \ar[r]  & \ \cdots \ .
}
\end{equation}
The connecting homomorphisms yield a sequence of maps 
\begin{equation}
\label{eq:connectinghoms}
H^{0}(M; \ZZ_2, \cU^1) \xrightarrow{ \ \partial_1 \ } H^{1}(M; \ZZ_2,
\cU^0) \xrightarrow{ \ \partial_2 \ } H^{2}(M; \ZZ_2, \cU^1)
\xrightarrow{ \ \partial_3 \ } H^{3}(M; \ZZ_2, \cU^0)
\xrightarrow{ \ \partial_4 \ } \cdots \ .
\end{equation}
Note that there is no reason to expect that $\partial_{k+1}
\circ \partial_k = 0$.  Indeed for $M = \pt$, the sequence becomes
\begin{equation}
\label{eq:connectinghomspoint}
\ZZ_2 \xrightarrow{ \ \partial_1 \ } \ZZ_2 \xrightarrow{ \ \partial_2
  \ }
\ZZ_2 \xrightarrow{ \ \partial_3 \ } \ZZ_2 \xrightarrow{ \ \partial_4
  \ } \cdots \
\end{equation}
and by the long exact sequences
\eqref{eq:beginlongsequence} it follows that these are all isomorphisms,
and in fact
equalities.  Since $\iota_m$ is a Real map,  the induced maps 
between the long exact sequences \eqref{eq:connectinghoms} and \eqref{eq:connectinghomspoint}
commute.  But $\sigma_{p+1}(\xi)(m) = \iota_m^*(\xi)$, so it follows that the connecting
homomorphisms commute with the sign choice maps and we have  

\begin{proposition}
\label{prop:sign-connecting}
$ 
\sigma_{p+2} \circ \partial_{p+2} = \sigma_{p+1}
$ \ 
for $p = -1, 0, 1, 2, \dots$.
\end{proposition}

In particular, this begins with the commutative diagram from Section~\ref{sec:Intro}:
\begin{equation}
\label{eq:3signs}
 \begin{gathered}   \xymatrix{
         H^0(M; \ZZ_2, \cU^1)     \ar@{->}[dr]^{ \sigma_0} \ar@{->}[d]^{ \partial_1}  &      \\
                H^1(M; \ZZ_2, \cU^0)        \ar@{->}[r]^{ \sigma_1}    \ar@{->}[d]^{ \partial_2}      &\ \ \ \  H^0(M^\tau, \ZZ_2)    \\
        H^2(M; \ZZ_2, \cU^1) \ar@{->}[ur]_{\sigma_2}\ar@{->}[d]^{ \partial_3}  &  \\
           \vdots  &  }
        \end{gathered}
\end{equation}

Let $M$ be a $d$-dimensional manifold.  Then $H^d(M, \cU(1)) =
H^{d+1}(M,\ZZ) = 0$, so  the long exact sequences
\eqref{eq:beginlongsequence} degenerate, and
it follows that $\partial_{d}$ is surjective and $\partial_k$ is an
isomorphism for $k > d$.  Let ${\mathit\Sigma}_p(M) = \im(\sigma_p)$ denote the image of the sign choice homomorphism and call its elements {\em geometric
   sign choices}. We conclude that 
$$
{\mathit\Sigma}_0(M)  \subseteq {\mathit\Sigma}_1(M) \subseteq \cdots \subseteq
{\mathit\Sigma}_{d-1}(M) = {\mathit\Sigma}_{d}(M) = {\mathit\Sigma}_{d+1}(M) = \cdots  \subseteq
H^0(M^\tau, \ZZ_2) \ .
$$

In the sequel we consider geometric constructions of sign choices in the particular 
cases of interest, and the general question of what sign choices on an orientifold can arise as geometric sign choices. 


\section{Real bundle gerbes and their connections}
\label{sec:realgerbes}

We are particularly interested in the higher
structures provided by invariant $p$-gerbes with $p=1$, as these are the
geometric data encoding the string orientifold backgrounds with NS--NS
$H$-flux. In this section we consider this case in more detail and describe how to properly incorporate the $B$-field into the picture.

\subsection{Real bundle gerbes}

We begin by recalling some basic facts about Real bundle
gerbes, referring the reader to~\cite{HekMurSza} for more details.

Let $M$ be a manifold and $Y \xrightarrow{\pi} M$ a surjective submersion. We denote by $Y^{[p]}$ the
$p$-fold fibre product of $Y$ with itself, that is, $Y^{[p]} = Y \times_M Y
\times_M \cdots \times_M Y$. This is a simplicial space $Y^{[\bullet]}$ whose face
maps are given by the projections $\pi_i \colon Y^{[p]} \to Y^{[p-1]}$
which omit the $i$-th factor. Let $\Omega^q(Y^{[p]})$ denote the space of differential $q$-forms on $Y^{[p]}$ and define 
$$
\delta \colon \Omega^q(Y^{[p-1]}) \longrightarrow \Omega^q(Y^{[p]})
\qquad\text{by}\quad \delta =  \sum_{i=1}^p\, (-1)^{i-1} \, \pi_i^* \ .
$$
These maps form an exact complex for all $q \geq 0$ called the {\em fundamental  complex} \cite{Mur}
$$
0 \longrightarrow \Omega^q(M) \xrightarrow{ \ \pi^* \ } \Omega^q(Y)
\xrightarrow{ \ \delta \ }  \Omega^q(Y^{[2]}) \xrightarrow{ \ \delta \
}
\Omega^q(Y^{[3]}) \xrightarrow{ \ \delta \ } \cdots \ .
$$
For any function $g \colon Y^{[p-1]} \to U(1)$, we define   $\delta(g)
\colon Y^{[p]} \to U(1)$ by $\delta(g) =\sum_{i=1}^p\, (-1)^{i-1}\, g \circ \pi_i$. Similarly, to every $U(1)$-bundle $P \to Y^{[p-1]}$ we associate a $U(1)$-bundle
$\delta(P) \to Y^{[p]}$ by
$$
\delta(P) = \pi_1^{-1}(P) \otimes \big(\pi_2^{-1}(P)\big)^* \otimes \pi_3^{-1}(P)
\otimes \cdots \ .
$$
One easily checks that $\delta(\delta(g)) = 1$ and that $\delta(\delta(P)) $ is canonically trivial.  

\begin{definition} A {\em bundle gerbe} $(P, Y)$ over $M$ is  a principal $U(1)$-bundle  $P \to Y^{[2]}$ together with a bundle gerbe multiplication defined by a bundle isomorphism over $Y^{[3]}$,
$$
 \pi_3^{-1}(P) \otimes \pi_1^{-1}(P) \longrightarrow \pi_2^{-1}(P) \ ,
$$
which is associative over $Y^{[4]}$. If  $P_{(y_1, y_2)}$ denotes the
fibre of $P$ over $(y_1, y_2)\in Y^{[2]}$, then associativity means that the diagram 
\begin{equation*}
\xymatrix{
P_{(y_1, y_2) } \otimes P_{(y_2, y_3) } \otimes P_{(y_3, y_4) } \ar[d]
\ar[r] & P_{(y_1, y_3) } \otimes P_{(y_3, y_4) } \ar[d] \\
P_{(y_1, y_2) } \otimes P_{(y_2, y_4) } \ar[r] &  P_{(y_1, y_4) }
}
\end{equation*}
commutes for all $(y_1, y_2, y_3, y_4) \in Y^{[4]}$.
\end{definition}

Bundle gerbes are classified up to stable isomorphism~\cite{MurSte2} by their
complex Dixmier--Douady class in the degree two sheaf cohomology
$$
\DD(P)\in H^2(M,\cU(1)) \ .
$$

Following \cite{Mou}, if $M$ is endowed with an involution $\tau \colon M \to M$ we define  a {\em Real structure} on a bundle gerbe $(P, Y)$ to be  a pair of maps $(\tau_P, \tau_Y)$, where $\tau_Y \colon Y \to Y$ is an involution 
covering $\tau \colon M \to M$, and $\tau_P \colon P \to P$ is a
 conjugate involution covering $\tau_Y^{[2]} \colon Y^{[2]}
\to Y^{[2]}$ and commuting with the bundle gerbe multiplication. Here
conjugate involution means that $\tau_P(p\, z) = \tau_P(p)\, \bar z$,
for $p\in P$ and $z\in U(1)$, and $\tau_P^2 = \id_P$. In the sequel we will suppress the subscripts on $\tau_P$ and~$\tau_Y$.

\begin{definition} A \emph{Real bundle gerbe} over $M$ is a bundle gerbe $(P, Y)$ over $M$ with a Real structure.
\end{definition}

In \cite{HekMurSza} we proved that Real bundle gerbes are classified up to Real stable isomorphism by their \emph{Real Dixmier--Douady class} in the degree two  equivariant sheaf cohomology group:
$$\DD_R(P)\in H^2(M; \ZZ_2, \cU^1 ) \ .$$ 

\subsection{Connective structures on Real bundle gerbes} 

Recall from \cite{Mur} that  if $(P, Y)$ is a bundle gerbe, a connection $A$ on the $U(1)$-bundle   $P \to Y^{[2]}$ is called  a {\em bundle
gerbe connection} if it respects the bundle gerbe multiplication. If $A$ is a bundle gerbe connection, then the curvature $F_A \in \Omega^2(Y^{[2]})$ satisfies $\delta(F_A) = 0$. Then the exactness of the 
fundamental complex implies that there exists $B \in \Omega^2(Y)$ such
that $F_A = \delta(B)$. This two-form is called a {\em curving} and the
pair $(A,B)$ is a {\em connective structure} on $(P,Y)$; in string
theory parlance, a connective structure is a $B$-field. As $\delta$
commutes with the exterior derivative $\dd$, we have  $\delta(\dd B) = \dd\delta(B) = 
\allowbreak \dd F_A  \allowbreak = 0$.  Hence $\dd B = \pi^*(H)$ for
some $H \in \Omega^3(M)$; in string theory the three-form $H$ is
called an $H$-flux. Moreover,
$\pi^*(\dd H) = \allowbreak \dd\pi^*(H) \allowbreak = \dd^2 B = 0$,
so the 3-curvature $H$ is closed and defines a de~Rham representative
for the complex Dixmier--Douady class of the bundle gerbe:
$$
\DD(P)  = \big[ \tfrac{1}{2\pi \ii}\, H\big] \ \in \ H^3(M, \ZZ) \ .
$$
    
Next let us consider the {\em dual bundle gerbe} $(P^*, Y)$. This is
defined by the same manifold $P$ but with the conjugate $U(1)$-action,
$(p,z) \mapsto  p\, \bar z$. If $A$ is a bundle gerbe connection on
$P$, then $- A$ yields a bundle gerbe connection on $P^*$. Indeed if
$\xi \in \ii\RR$, let $\iota_p(\xi) \in T_p P $ be the vertical vector
generated by the action of $U(1)$, and similarly let $\iota_p^*(\xi)
\in T_p P^* $. Since $\iota^*_p(\xi) =   - \iota_p(\xi)$ and
$A(\iota^*_p(\xi)) = - \xi$, we conclude that $-A$ is a connection on
$P^*$ with curvature $- F_A$. For  any choice of curving $B \in
\Omega^2(Y)$ for $A$, it follows that $- B$ is a  curving for $-A$ and the
3-curvature of the connective structure $(-A, -B)$ for $(P^*, Y)$ is given by $- H$. 

Let $(M, \tau)$ be an orientifold and  $(P, Y)$ a Real bundle gerbe 
over $M$. Then $\tau \colon P^* \to P $
is a bundle gerbe isomorphism.  
\begin{definition} A {\em Real connective structure} for $(P,Y)$ or
  \emph{orientifold $B$-field} is a
  $B$-field $(A,B)$ satisfying $\tau^*(A,B) = (-A,-B)$.
\end{definition}
\begin{proposition} Real connective structures for $(P,Y)$ exist.
\end{proposition}
\begin{proof}
The existence of bundle gerbe connections is shown in \cite{Mur}.
Real bundle
gerbe connections exist because if $a$ is any bundle gerbe connection, then 
$$
A = \tfrac{1}{2}\, a +  \tfrac{1}{2}\, \big( -\tau^*( a)\big)
$$
is a Real bundle gerbe connection. Any convex combination of
connection one-forms is also a connection one-form.  It follows that the curvature $F_A \in \Omega^2(Y^{[2]})$ 
also satisfies $\tau^*(F_A) = - F_A$.  Consider a choice of curving $B \in \Omega^2(Y)$
 satisfying $\delta(B) = F_A$. Hence 
$$
\delta\big(\tau^*( B)\big) = \tau^*(F_A) = - F_A = \delta(- B) \ .
$$
 It follows that $\delta( B + \tau^*( B)) = 0$ and thus there exists
 $\psi \in \Omega^1(M)$ such that $B + \tau^*( B) =  \pi^*(\psi)$.
 Applying the involution we obtain $\tau^*(B) + B =  \tau^*\, \pi^*(\psi)$, and since pullback of forms along $\pi$ is injective, we find $\tau^*( \psi) =  \psi$. Hence
 $$
 B - \tfrac{1}{2}\, \pi^*(\psi) = - \tau^*( B) + \tfrac{1}{2}\,
 \pi^*\big(\tau^*(\psi)\big) = - \tau^*\big(B - \tfrac{1}{2}\,
 \pi^*(\psi) \big) \ ,
 $$
and moreover 
$$
F_A = \delta\big( B - \tfrac{1}{2}\, \pi^*(\psi)\big)
$$
because $\delta \circ \pi^* = 0$. 
We conclude that a curving $B$ for a Real bundle gerbe connection can always be chosen such that $\tau^*(B) = - B$.
\end{proof}  
If $(A, B)$ is an orientifold $B$-field, it follows that the $H$-flux satisfies 
$\tau^*(H) = -  H$, so it gives a de Rham representative for the Real Dixmier--Douady class in the Borel equivariant cohomology with local coefficients $H^3(E\ZZ_2\times_{\ZZ_2}M,\ZZ(1))$~\cite{HekMurSza}:
$$
\DD_R(P) =  \big[ \tfrac{1}{2\pi \ii}\, H\big] \ \in \ H^3_{\ZZ_2}(M, \ZZ(1)) \ .
$$ 

\begin{example}[Real bundle gerbes on $S^1\times S^2$]
If $G$ is a compact simple Lie group with maximal torus $T$ then there
is a map, called the Weyl map, $T \times G/T \to G$ defined
by $(t, g\,T) \mapsto g\,t\,g^{-1}$ for $t\in T$ and $g\in G$. This 
is a finite covering with fibre the Weyl group over the open dense
subset of $G$ of regular elements.  In particular, if $G = SU(2)$ we
get a map
$S^1 \times S^2 \to S^3$ which is a double covering over $S^3$ minus the north and south poles which are the points $I$ and $-I$ in $SU(2)$, the only 
elements of $SU(2)$ whose eigenvalues are not distinct.  
 Under this map the basic bundle gerbe on $S^3$, whose complex Dixmier--Douady
 class generates $H^3(S^3,\ZZ)=\ZZ$, pulls back to 
a bundle gerbe on $S^1 \times S^2$ whose complex Dixmier--Douady class is twice a generator in $H^3(S^1 \times S^2, \ZZ) = \ZZ$.  This is because
its integral over $S^3$ is equal to $1$ but the Weyl map is two-to-one
on a dense open subset, so the integral of the pulled back form is
equal to $2$. 

The generator of $H^3(S^1 \times S^2, \ZZ) = \ZZ$ is given by a bundle
gerbe which is a cup product of a function and a line bundle; we call
it the cup-product bundle gerbe on $S^1\times S^2$. 
Let $L \to S^2$ be the Hopf bundle; the pullback of its Chern class
generates $H^2(S^1\times S^2,\ZZ)=\ZZ$. Let $Y = \RR \times S^2 \to S^1
\times S^2$ be the covering space where $S^1 = \RR / \ZZ$. Then $Y^{[2]}$ can be identified
with the set of points $(x, y, u) \in \RR \times \RR \times U(1)$ where $ x - y \in \ZZ$.  Define $P \to Y^{[2]} $ by $P_{(x, y, u)} = L^{x-y}_u$.  The bundle gerbe multiplication is immediate from 
$$
P_{(x, y, u)} \otimes P_{(y, z, u)} =L^{x-y}_u \otimes L^{y-z}_u
= L^{x-z}_u = P_{(x, z, u)} \ .
$$
It is straightforward to construct a connection and curving whose
$3$-curvature is $H=\dd(2\pi \ii \theta) \wedge \vol_{S^2}$, where $\theta\in [0,1)$ is a
local coordinate on $S^1$ and $\vol_{S^2}$ is the unit area form on $S^2$. 

The Weyl map in this case also has a simple description as 
\begin{align}\label{eq:WeylSU2}
S^1 \times S^2 \longrightarrow S^3 \ , \qquad 
\big((a, b) \,,\, (x, y, z) \big) \longmapsto (b\,x, b\,y, b\,z,a) \ .
\end{align}
There is no line bundle on $S^3$ which pulls back to the Hopf bundle on
$S^1\times S^2$ because all line bundles on $S^3$ are trivial.
We give the three-sphere $S^3$ the Real structure $(x, y, z,w) \mapsto (-x, -y, -z,w)$ making it $S^{3,1}$. 
It is then immediate from \eqref{eq:WeylSU2} that there are  two lifts of this Real structure to $S^1 \times S^2$: either to $\big((a, b) , (x, y, z) \big) \mapsto \big((a, -b) , (x, y, z) \big)$
making it 
$S^{1, 1} \times S^{0,3}$, or to $\big((a, b) , (x, y, z) \big) \mapsto \big((a, b) , (-x, -y, -z) \big)$ making it $S^{0,2} \times S^{3,0}$. 
Here $S^{1, 1}$ is the circle $S^1$ with the conjugation action on
complex numbers of modulus one and $S^{3,0}$ is the two-sphere $S^2$ with  the 
antipodal map.  These are quite different Real structures; on $S^{0,2} \times S^{3,0}$
the involution is free whereas the involution on $S^{1, 1} \times
S^{0,3}$ gives two orientifold planes $\{ \pm\, 1 \} \times S^2
$. By~\cite{HekMurSza} the only line bundes on $S^2$ which admit a
Real structure are those with even Chern class. Then using either the long exact sequence \eqref{eq:beginlongsequence} or
the spectral sequence for the Grothendieck equivariant cohomology one
can characterise the Real bundle gerbes in these two cases. The
calculations are straightforward and so we only summarise the
results:
\begin{itemize}

\item The Real Brauer group of  $S^{0,2} \times S^{3,0}$ equals 
$\ZZ$ and the Real Dixmier--Douady class is given by $n \mapsto
2n$. So the pullback of the basic bundle gerbe on $S^3$
whose complex Dixmier--Douady class is twice the generator of $H^3(S^1 \times S^2,\ZZ) $ admits a Real structure, but the 
cup-product bundle gerbe whose complex Dixmier--Douady class is the generator and whose $H$-flux is $H =  \dd(2\pi \ii\theta) \wedge \vol_{S^2}$,  satisfying the necessary condition $\tau^*(H) = -H$, does not.

\item The Real Brauer group of $S^{1, 1} \times S^{0,3}$ equals 
$\ZZ_2 \oplus \ZZ_2 \oplus \ZZ_2 \oplus \ZZ$ and the Real Dixmier--Douady class is given by $(n,m,k,l) \mapsto l$. So the 
cup-product bundle gerbe whose complex Dixmier--Douady class is the generator and whose $H$-flux is $H = \dd(2\pi \ii \theta) \wedge \vol_{S^2}$,
satisfying $\tau^*(H) = -H$,  admits a Real structure. 

\end{itemize}
\end{example}

\subsection{Holonomy of orientifold $B$-fields} 
\label{sec:Bfieldampl}
Let $(P, Y)$ be a bundle gerbe over  a closed oriented surface $\Sigma$ and choose a connective structure $(A, B)$.  As $H^3(\Sigma, \ZZ)= 0$, the bundle gerbe is trivial.  Fix a trivialisation $R \to Y$ with isomorphism $\delta(R) \simeq P$ and let $a $ be a connection on $R$.  Then $A $ and $\delta(a)$ are both bundle gerbe connections for $(P, Y)$.  Their difference is a one-form $A -\delta(a)  \in \Omega^1(Y^{[2]})$ satisfying $\delta(A - \delta(a)) = 0$. Hence there is a one-form $\a$ on $Y$ 
satisfying  $A = \delta(a) + \delta(\a) = \delta(a + \a)$.  In other words, we may choose a connection $a$ 
on $R$ such that $\delta(a) = A$.  It follows that $\delta(F_a) = F_A = \delta(B)$ and so 
$B - F_a = \pi^*(\mu_a)$ for some $\mu_a \in \Omega^2(\Sigma)$. If we
change the connection $a$ to $b = a + \pi^*(\delta(\rho))$ for some $\rho\in\Omega^1(\Sigma)$, then $B - F_b = B - F_a - \dd\pi^*(\rho)$ so that 
$$
\mu_a = \mu_b + \dd \rho\ .
$$
The {\em holonomy} of $(A, B)$ over $\Sigma$ is defined by 
$$
\hol(\Sigma; A, B) = \exp \Big( \int_\Sigma \, \mu_a  \Big) \ \in \
U(1)\ .
$$
We note that the integral depends on the orientation of $\Sigma$, but is independent of the choice of $a$. 
Any other trivialisation of the bundle gerbe is of the form $R \otimes \pi^*(T)$ for 
some line bundle $T \to \Sigma$.  The connection $a$ tensored by the pullback of a connection 
$\nabla$ on $T$ changes $\mu_a$  to $\mu_a - F_\nabla$, so the holonomy remains unchanged, 
$$
\exp \Big( \int_\Sigma \, \mu_a  \Big) \, \exp \Big(- \int_\Sigma \, F_\nabla \Big)
= \exp \Big( \int_\Sigma \, \mu_a  \Big) \ ,
$$
since the curvature $F_\nabla$ is an integral two-form. Hence the holonomy  depends only on the triple $(\Sigma; A, B)$.

Now let $(A,B)$ be a $B$-field on $M$. For any smooth map $\phi \colon \Sigma \to M$, we define the {\em holonomy of $(A, B)$ along $\phi$} by
$$
\hol(\phi ; A, B) = \hol\big(\Sigma; \phi^*(A), \phi^*(B)\big) \ .
$$
It is straightforward to check that the holonomies of connective structures for a bundle gerbe and its dual are related by 
$$
\hol(\phi; A, B) = \hol( \phi; -A, -B)^{-1} \ .
$$
Hence if $(A, B)$ is a Real connective structure for a Real bundle gerbe $(P,Y)$, we conclude that
$$
\hol(\phi; A, B) = \hol( \tau \circ \phi; -A, -B)^{-1} \ .
$$

In applications to the orientifold constructions of type~II string
theory, the $\ZZ_2$-action on the spacetime $M$ generated by the
involution $\tau$
is combined with orientation-reversal of the string worldsheet. In this case the surface
$\Sigma$ is not oriented and need not even be orientable, and the string
fields $\phi$ are smooth maps from $\Sigma$ to the quotient space
$X=M/\tau$ which represents the ``physical points'' of the orientifold
spacetime. This is achieved in the orientifold sigma-model by regarding $\phi$ as maps to the total
space of the fibration $M\to X$, and gauging the symmetry of the
string theory.
For this, we introduce the orientation double
cover $\widehat{\Sigma}\to\Sigma$ corresponding to the first
Stiefel--Whitney class $w_1(\Sigma)\in H^1(\Sigma,\ZZ_2)$, which is
canonically oriented with a canonical 
orientation-reversing involution
$\Omega:\widehat{\Sigma}\to\widehat{\Sigma}$ permuting the sheets and
preserving the fibres, called worldsheet parity. The string fields are then smooth maps
$\hat\phi:\widehat{\Sigma}\to M$ which are equivariant,
$$
\hat\phi\circ\Omega = \tau\circ\hat\phi \ .
$$
In \cite{SSW, GSW} they show that in this situation it is possible to define
the 
{\em unoriented surface holonomy} of a Jandl gerbe which is a 
square root of $\hol\big(\widehat{\Sigma};\hat\phi{}^*(A),\hat\phi{}^*(B)\big)$.
We may define the  \emph{orientifold $B$-field holonomy} $
\hol(\phi;A,B)$ of a Real bundle gerbe with Real connection by noting that every Real bundle gerbe  
is naturally a Jandl gerbe and applying the construction in \cite{SSW, GSW}. To see why this is true, note that in our language a Jandl structure on a gerbe $(P, Y)$
is a $\ZZ_2$-equivariant trivialisation of the natural $\ZZ_2$-action on $P\otimes \tau^{-1}(P)$.  For a Real bundle gerbe $(P, Y)$ this 
arises because
$P \otimes \tau^{-1}(P) \simeq \d L$
for $L_{(y, x)} = P^*_{(y, \tau( x))}$. The required switch isomorphism of  \cite{SSW}  $\hat s \colon P^*_{(x, \tau(y))} \to P^*_{(y, \tau(x))}$ is induced by $\tau \colon P_{(x, y)} \to P^*_{(\tau(x), \tau(y))}$.  It is a straightforward exercise to check that this satisfies
the required conditions to be a Jandl structure, so we can apply the results of \cite{SSW,GSW} to define $\hol(\phi;A,B)$. 
This defines precisely the $B$-field amplitude which was only
schematically discussed in~\cite{HekMurSza}. By \cite{SSW,GSW} it is
invariant under the combined actions of the involutions $\Omega$ and
$\tau$ which define the string orientifold construction and satisfies 
$$
\hol(\phi;A,B) =
\sqrt{\hol\big(\widehat{\Sigma};\hat\phi{}^*(A),\hat\phi{}^*(B)\big)}
\ .
$$

As a further application, we can turn the construction of Real
holonomy around and use it to provide sufficient conditions for the existence of Real structures on bundle gerbes over $(M,\tau)$. 
  
  \begin{proposition}\label{tautologicalgerbe}
  Let $(M, \tau)$ be an orientifold, where $M$ is $2$-connected, $M^\tau
  \neq \emptyset$ and $H$ is a three-form on $M$ with integral periods
  satisfying $\tau^*(H) = - H$.  Then there is a tautological Real bundle gerbe over $M$ with complex Dixmier--Douady class
  $\big[\frac1{2\pi\ii}\,H\big] \in H^3(M,\ZZ)$. 
  \end{proposition}
  \begin{proof}
First we recall the construction of a bundle gerbe $(P,Y)$ from an
integral three-form  $H$  on a $2$-connected manifold $M$.  Fix  a
basepoint $m \in M$ and let  $Y = \cP M$ be the space
of paths based at $m$  with endpoint evaluation as  projection to $M$.
If $p_1, p_2\in Y$ have the same endpoint, choose an oriented surface
$\Sigma$ in $M$
spanning them, that is, the boundary of $\Sigma$ is $p_1$ concatenated
with the oppositely oriented path $p_2$.  Then the fibre of $P \to
Y^{[2]}$ consists of all triples $\big((p_1, p_2), \Sigma, z\big)$, with $z\in
U(1)$, modulo the equivalence relation 
$\big((p_1, p_2), \Sigma, z\big) \equiv \big((p_1, p_2), \Sigma', z'\big)$
if   $\hol(\Sigma \cup \Sigma'; H)\, z = z'$. Here $\Sigma\cup\Sigma'$
is the closed surface obtained by gluing $\Sigma$ to the oppositely
oriented $\Sigma'$ along their common boundary, while $\hol(S; H)$ for
any closed oriented surface $S\subset M$ is the holonomy expressed as the standard Wess--Zumino--Witten term 
$$
\hol(S; H)  = \exp\Big( \int_{B_S}\, H \Big) \ ,
$$
where $B_S$ is a three-manifold  whose boundary is $S$. This is well-defined since $H$ is an integral form. 

Next we consider a Real version of this construction. There is a natural way to pair elements of $P_{(p_1, p_2)}$ with
elements of $P_{(p_2, p_1)}$.  Given $\big((p_1, p_2), \Sigma,
z\big)$, it can be paired with $\big((p_2, p_1), \Sigma, w\big)$ by
keeping the same surface $\Sigma$ but changing its orientation.  Then
we declare these two to pair to give $z\, w\in U(1)$.  More generally, we can pair $\big((p_1, p_2), \Sigma, z\big)$ and 
$\big((p_2, p_1), \Sigma', w\big)$ by first using the equivalence relation to replace $\Sigma'$ with $\Sigma$.

Now assume that the basepoint $m$ of $M$ is a fixed point of $\tau$, so that $\tau$ acts on everything.  Assume further that $\tau^*(H) = -H$.  We define a Real structure  $\tau$ by showing that 
$$
 \big((p_1, p_2), \Sigma, z\big) \longmapsto \big(( \tau(p_1), \tau(p_2)), \tau(\Sigma), \bar z\big)
 $$
descends through the equivalence relation to give a conjugate bundle gerbe isomorphism $P \to \tau^{-1}(P)$. 

The map $\tau$ satisfies $\tau^2 = 1$. If $\big((p_1, p_2), \Sigma, z\big) \equiv \big((p_1, p_2), \Sigma', z'\big)$, then $\hol(\Sigma \cup \Sigma'; H)\, z = z' $.
Now consider $\big((\tau(p_1), \tau(p_2)), \tau(\Sigma), \bar z\big)$
and $\big((\tau(p_1), \tau(p_2)), \tau(\Sigma'), \overline{z'}\,\big)$.
Then we find 
\begin{align*}
\hol\big( \tau(\Sigma) \cup \tau(\Sigma'); H\big) 
&= \hol\big(\tau(\Sigma) \cup \tau(\Sigma'); -\tau^*(H) \big)\\[4pt]
&= \hol(\Sigma \cup\Sigma'; -H) \\[4pt]
&= \overline{\hol(\Sigma \cup\Sigma'; H )} \ ,
\end{align*} 
so that $\hol\big(\tau(\Sigma) \cup\tau(\Sigma'); H\big)\, \bar z = \overline{z'}$ as required.

Finally, we check that $\tau$ commutes with the bundle gerbe
multiplication. On the one hand, the multiplication is given by
$$
\big((p_1, p_2),\Sigma, z\big) \, \big((p_2, p_3),\Sigma', z'
\big) = \big((p_1, p_3),\Sigma \cup\Sigma', z\, z' \big) \ .
$$
On the other hand, under $\tau$ we obtain
$$ 
\big((\tau(p_1),\tau(p_2)),\tau(\Sigma), \bar z\big) \, \big((\tau(p_2),\tau(p_3)),\tau(\Sigma'), \overline{z'}\, \big) = \big((\tau(p_1),\tau(p_3)), \tau(\Sigma \cup\Sigma'), \overline{z\, z'}\, \big)
$$
as required.
\end{proof}

\begin{remark}
In general, the class of the three-form $H$ of Proposition~\ref{tautologicalgerbe} does not lift to the Real
Dixmier--Douady class of the tautological Real bundle gerbe in Borel
equivariant cohomology with local coefficients. The conditions under
which it does, and hence represents the $H$-flux of an orientifold $B$-field, are discussed in~\cite{HekMurSza}.
\end{remark}


\section{Geometric constructions of sign choices}
\label{sec:bgsign}

In this section we explain how the sign choice homomorphism
\eqref{eq:signchoice} can be constructed geometrically from the
invariant $p$-gerbes when $p = -1, 0, 1$, that is, for Real functions, $\ZZ_2$-equivariant line bundles and Real bundle gerbes.

\subsection{Sign choice for Real functions}
 
Let $g$ be a Real function on $(M,\tau)$. If $m\in M^\tau$ then $g(m) = \overline{g(m)} = \pm\, 1$, so there is a natural sign choice homomorphism
 $$
\sigma_0 \colon H^0(M; \ZZ_2, \cU^1) \longrightarrow H^0(M^\tau, \ZZ_2)
$$
defined by evaluation $\sigma_0(g)(m) = g(m)$. This map is
multiplicative in the sense that $\sigma_0(g\,h) = \sigma_0(g)\,
\sigma_0(h)$. If  $g = f \, ( \bar f \circ \tau ) $ for an arbitrary
function $f$ on $M$ and $m$ is a fixed point of
 $\tau$, then $g(m) = f(m) \, \overline{ f(m)} = 1$ so that
 $\sigma_0(f\, (\bar f \circ \tau)) = 1$. In other words, the composition 
 $$
 H^0(M, \cU(1)) \xrightarrow{f \mapsto f\, (\bar f \circ \tau) }  H^0(M; \ZZ_2, \cU^1) 
 \xrightarrow{ \ \sigma_0 \ } H^0(M^\tau, \ZZ_2)
 $$
 is equal to $1$. In fact, there is a stronger result in this case
 given by
\begin{proposition}
\label{prop:Hzero}
If $M$ is $1$-connected, then the sequence 
$$
0 \ \longrightarrow \   H^0(M; \ZZ_2, \cU^0) \ \longrightarrow \   H^0( M, \cU(1)) \ \xrightarrow{ f \mapsto f\, (\bar f \circ \tau)} \   H^0(M; \ZZ_2,  \cU^1) \ \xrightarrow{ \ \epsilon \ } \ \ZZ_2 \ \longrightarrow \ 0
$$
is exact.
\end{proposition}
\begin{proof}
If $g \colon M \to U(1)$ is a Real function, consider a lift $\hat g
\in H^0( M, \cU(1))$. Then $\hat g \circ \tau + \hat g = k$ for some
$k\in\ZZ$, and the image $[k]\in\ZZ_2$ is well-defined independently
of the lift $\hat g$ of $g$. We call it $\epsilon(g)$. 
If $\epsilon(g) = 0$, then we can define $\hat f = -\frac{1}2\,\hat g$ with 
$$
\hat f \circ \tau   -  \hat f  =  \tfrac{1 }2\,\hat g - \tfrac{1}2\,(\hat g \circ \tau) = \hat g
$$ 
so that $(f \circ \tau)\, \bar f = g $. 
\end{proof}

\begin{corollary} If $M$ is $1$-connected, then ${\mathit\Sigma}_0(M)\subseteq\ZZ_2$.  
\end{corollary}
Hence if $M$ is $1$-connected and $M^\tau$ contains more than one O-plane, then $\sigma_0$ is not surjective. We discuss further constraints on geometric sign choices in Section \ref{sec:constraining}.

\subsection{Sign choice for equivariant bundles}
\label{sec:signbundle}

Let us first revisit a  basic result from~\cite{HekMurSza}, namely the
problem of classifying, up to isomorphism, the $\ZZ_2$-structures on
a $U(1)$-bundle $L$ which admits a $\ZZ_2$-structure.

If $\tau$ is a $\ZZ_2$-structure on $L$, then any other
$\ZZ_2$-structure $\hat\tau $ is given by $\hat\tau = g \, \tau$ for
some function $g \colon M \to U(1)$. 
By $\hat\tau{}^2 = 1$ it follows that $g\, ( g \circ \tau )  = 1$ or $g = \bar g \circ \tau$.
We say that $\hat\tau $ and $\tau$ are isomorphic
if there is a smooth function $f \colon M \to U(1)$ such that
$\hat\tau\, f = f\, \tau$; this is equivalent to $(f\circ \tau)\, g\,
\tau = f \, \tau$ or $g = f \, (\bar f \circ \tau ) $.  
\begin{proposition}
\label{prop:Z2-structures}
The space of non-isomorphic $\ZZ_2$-structures on a bundle $L \to M$ admitting a $\ZZ_2$-structure is the quotient of $H^0(M; \ZZ_2, \cU^1)$ by the image of the homomorphism
$$
H^0(M;  \cU(1)) \xrightarrow{f \mapsto f\, (\bar f \circ \tau) }
H^0(M; \ZZ_2, \cU^1) \ .
$$
\end{proposition}

From  this description it follows that the space of $\ZZ_2$-structures on $L$ is independent of the choice of $L$ as long 
as $L$ admits at least one $\ZZ_2$-structure. 

If $L$ is a $\ZZ_2$-equivariant line bundle over $(M,\tau)$, then its
fibre over a fixed point $m\in M^\tau$ is simply a copy of $\CC$
with a $\ZZ_2$-action by multiplication with $\pm\, 1$ which we label $\sign(L_m)$. This allows us to define a homomorphism
 $$
 \sigma_1 \colon H^1(M; \ZZ_2, \cU^0) \longrightarrow H^0(M^\tau, \ZZ_2)
$$
by $\sigma_1(L)(m) = \sign(L_m)$. Again it is clear that this map is multiplicative in the sense that $\sigma_1(L \otimes K) = \sigma_1(L)\, \sigma_1(K)$. If $L$ is an arbitrary line bundle on $M$,
then $L \otimes \tau^{-1}(L)$ has the canonical $\ZZ_2$-action induced by swapping factors
of the tensor product 
$$
L_x \otimes L_{\tau(x)} \longrightarrow L_{\tau(x)} \otimes L_x \ .
$$
At a fixed point this is the identity map, so it has sign choice equal to $1$. It follows that the composition
 $$
 H^1(M, \cU(1)) \xrightarrow{ L \mapsto L \otimes \tau^{-1}(L) }  H^1(M; \ZZ_2, \cU^0) 
 \xrightarrow{ \ \sigma_1 \ } H^0(M^\tau, \ZZ_2)
 $$
 is equal to $1$.

 \subsection{Sign choice for Real bundle gerbes}
 \label{sec:general}

 If $(P, Y)$ is a bundle gerbe and $y \in Y$,  then the fibre $P_{(y, y)} = U(1)$. 
 This follows from the fact that the  multiplication map $ P_{(y, y)}  \otimes P_{(y, y)} \to P_{(y, y)} $ commutes
 with the $U(1)$-action on $P_{(y, y)}$. 
So if $p \in P_{(y, y)} $ and $q \in P_{(y, y)}$, then $p\,q \in
P_{(y, y)} $ and hence $p\, q = p \, \< q \>$ for some $\< q \> \in U(1)$. 
If we change $p $ to $p\,z$ for $z \in U(1)$, then 
$$
(p\,z)\,q = (p\,q)\,z = p\,\< q \>\, z = (p\,z)\,\< q \>
$$ 
so the map $q \mapsto \< q \>$ is independent
of the choice of $p$.  We need two results about this map.  Firstly, if $r \in P_{(y', y)}$ then 
$$
(r\, p\, r^{-1})\, (r \, q\, r^{-1}) = r\, (p\, q)\, r^{-1} = r\,
p\,\< q \>\, r^{-1} = (r\, p\, r^{-1})\, \< q \>
$$ 
so that $\< r\, q\, r^{-1} \> = \< q \>$.
Secondly if $\tau$ is a Real structure, then $\tau(p\,q) =
\tau(p)\,\tau(q)$ and $\tau(p\,q) = \tau(p\, \< q \>) = \tau(p)
\, \overline{ \< q \> }$,
so that $\< \tau(q) \> = \overline{\< q \>}$.

Consider now a Real bundle gerbe $\cG = (P, Y)$ over $M$, and let $m
\in M^\tau$ be a fixed point of $\tau$.  Choose $y \in Y_m$, then $(y, \tau(y)) \in 
Y^{[2]}$.  Let $f \in P_{(y, \tau(y))} $ so that $f\, \tau(f) \in P_{(y, y)}$, and define
\begin{equation}
\label{eq:sigma}
\sigma_2(\cG)(m) = \sigma_2(P, Y)(m) = \< f \, \tau(f) \> \ .
\end{equation}

 \begin{proposition}
 \label{signchoice}
\begin{enumerate}[(a)]
\item $\sigma_2(\cG)(m)$ is independent of the choice of $f$ and $y$.
\item $\sigma_2(\cG)(m) \in \ZZ_2$.
\end{enumerate}
\end{proposition}
 \begin{proof}
(a) If we change $f $ to $f \, z $ for $z \in U(1)$, then 
 $$
\<(f\,z)\, \tau(f\,z) \>  = z\, \bar z\, \<f\, \tau(f)\> = \<f\,
\tau(f)\> \ ,
$$ 
 so  $\< f \, \tau(f)\> \in U(1)$ is independent of the choice of $f \in P_{(y, \tau(y))}$.
 Changing $y$ to $\tilde{y} \in Y_m$ and letting $g \in P_{(\tilde{y},
   y)}$, then $g \, f\, \tau(g^{-1})  \in P_{(\tilde{y}, \tau(\tilde{y}))}$ and 
 $$
\big\< \big( g\, f\, \tau(g^{-1})\big)\, \tau\big( g\, f\,
\tau(g^{-1}) \big)\big\> = \big\<g\, f\, \tau(g^{-1})\, \tau(g)\,
\tau(f)\, g^{-1} \big\> = \big\< g\, \big( f\, \tau(f) \big)\, g^{-1}\big\>
 =\<f\, \tau(f) \> \ .
 $$

(b) From (a) we know that $\<\tau(f)\, f\> = \<f\, \tau(f)\>$ so  
 $$
\overline{\<f\, \tau(f)\>}  
 = \big\<\tau\big( f\, \tau(f)\big) \big\> = \<\tau(f)\, f\> = \< f\, \tau(f)\> \ .
$$
Hence $\sigma_2(\cG)(m) \in \ZZ_2$.
 \end{proof} 
 
 This shows that letting $m$ vary determines an element $\sigma_2(\cG) \in
 H^0(M^\tau, \ZZ_2)$.  The next result implies that the sign choice of a Real bundle
 gerbe depends only on its Real stable isomorphism class.
  
  \begin{proposition} 
  \label{prop:signstable}
  \begin{enumerate}[(a)]
  \item If $\cG$ is a Real bundle gerbe, then $\sigma_2(\cG^*) =  \sigma_2(\cG)$. 
  \item If $\cG_1$ and $\cG_2$ are Real bundle gerbes, then $\sigma_2( \cG_1 \otimes \cG_2 ) = \sigma_2( \cG_1 )\, \sigma_2(  \cG_2 )$.
  \item If $\cG$ is a trivial Real bundle gerbe, then $\sigma_2(\cG) = 1$.
  \item If $\cG_1$ and $\cG_2$ are Real stably isomorphic Real bundle gerbes, then $\sigma_2(\cG_1) = \sigma_2(\cG_2)$.
   \end{enumerate}
  \end{proposition}
  \begin{proof}
 It is straightforward from the definition \eqref{eq:sigma} to prove (a) and (b). Part (d) follows from (c)  because
  $\cG_1$ and $\cG_2$ are Real stably isomorphic if and only if $\cG_1 \otimes \cG_2^*$ is Real trivial.  For~(c),
 take a Real bundle $R \to Y$ and for $(y_1,y_2)\in Y^{[2]}$ define $P_{(y_1, y_2)} = R_{y_2} \otimes R_{y_1}^*$ with the Real 
  structure induced by that of $R$.  Then the bundle gerbe multiplication 
  $$
  P_{(y_1, y_2)} \otimes P_{(y_2, y_3)} \longrightarrow P_{(y_1, y_3)}
  $$
  comes from the obvious contractions of 
  $$
  R_{y_2} \otimes R_{y_1}^* \otimes R_{y_3} \otimes R_{y_2}^*
  \longrightarrow R_{y_3} \otimes R_{y_1}^* \ .
  $$
It is straightforward to show that $\< r \otimes s^* \> = s^*(r)$ for  $r\otimes s^*\in P_{(y, y)} = R_y \otimes R_y^*$.
  
Let  $r \in R_y$ and define 
  $f = \tau(r)\otimes r^* \in  P_{(y, \tau(y))}=R_{\tau(y)}  \otimes R^*_{y}$, where $r^*$ is the element that satisfies $r^*(r) = 1$.  Then
  \begin{align*}
 \< f \, \tau(f) \>  & = \big\< \big(\tau(r) \otimes r^*\big)\, \big(r \otimes
                       \tau(r^*)\big)\big\> = \<\tau(r)\otimes\tau(r)^*\> = 1
                 \end{align*}
 as required.
  \end{proof} 

Thus we have a well-defined homomorphism
 $$
 \sigma_2 \colon  H^2(M; \ZZ_2, \cU^1 )  \longrightarrow H^0(M^\tau,
 \ZZ_2) \ .
 $$

\begin{example}[Real bundle gerbes over a point]
\label{ex:bg-point}
If $M = \pt$, then there are two Real stable isomorphism classes of
Real bundle gerbes (cf.\ Section~\ref{sec:signs}). 
Their sign choices are $+1$ for the trivial one and $-1$ for the
non-trivial one.  By Proposition~\ref{prop:signstable}~(c) it is clear that the trivial 
bundle gerbe has sign choice $+1$. We construct the non-trivial sign choice as follows. Let $M = \{m\}$ and $Y = M \times \ZZ_2$ 
with the Real structure $\tau(m, 0) = (m, 1)$ and $\tau(m, 1) = (m, 0)$. 
Then $Y^{[2]} = M \times \ZZ_2 \times \ZZ_2$ and we take the trivial bundle $P \to Y^{[2]}$ equipped with the trivial
multiplication, so that $p \mapsto \< p \>$ is just the identity map.
We define the Real structure  to be conjugation  from $P_{((m, 0),(m,0))} \to P_{((m, 1),(m,1))}$
and $P_{((m, 1),(m,1))} \to P_{((m, 0),(m,0))}$, and $-1$ times conjugation from $P_{((m, 0),(m,1))} \to P_{((m, 1),(m,0))}$
and $P_{((m, 1),(m,0))} \to P_{((m, 0),(m,1))}$. 

A straightforward calculation shows that this commutes with the bundle gerbe
multiplication and squares to the identity map.   To calculate the sign, let $y = (m, 0)$ so that $\tau(y) = (m, 1)$ and let $f \in P_{((m, 0),(m,1))}$ be $
1$.   Then $\< f\, \tau(f) \>  = -1$ as required. This bundle gerbe cannot be a trivial Real bundle gerbe
 because it has a non-trivial sign.  The fact that $H^2(\{m \}; \ZZ_2, \cU^1 ) = \ZZ_2$ shows that these are the 
 only two Real stable isomorphism classes of Real bundle gerbes over a point, and the calculation demonstrates that 
$$
 \sigma_2 \colon  H^2(\{m \}; \ZZ_2, \cU^1 )  \longrightarrow H^0(\{m \}^\tau, \ZZ_2)
 $$ 
is an isomorphism. In Section~\ref{sec:GaoHori} we will generalise this construction and show that it agrees with the connecting homomorphism  in \eqref{eq:connectinghoms}.
\end{example}

\begin{example}[Tautological Real bundle gerbes]
\label{ex:tautgerbe}
Under the conditions of Proposition~\ref{tautologicalgerbe}, we show
that the sign choice of the tautological Real bundle gerbe $(P,Y)$ is
given by 
$$
\sigma_2(P) = \hol\big(\Sigma \cup\tau(\Sigma); H\big) \ . 
$$
Let $m\in M^\tau$ be the basepoint in the proof of
Proposition~\ref{tautologicalgerbe} and let $x \in M^\tau$ be another
fixed point.  Choose $p_1 = p$ to be a path from $m$ to $x$  and $p_2
= \tau(p)$ as per the way one computes the sign choice, see
Proposition \ref{signchoice}. Picking a spanning surface $\Sigma$, we
then have
$[p,\tau(p),\Sigma, 1] \in P_{(p,\tau(p))}$ and 
$$
\tau\big([p,\tau(p),\Sigma, 1]\big) = \big[\tau(p), p,\tau(\Sigma), 1
\big] = \big[\tau(p), p,\Sigma, \hol\big(\Sigma \cup\tau(\Sigma);
H\big) \big] \ . 
$$
These pair to give the sign choice  $\hol\big(\Sigma \cup\tau(\Sigma);
H\big). $ This construction will be used later to exhibit a constraint
imposed on the sign choices when $M$ is $2$-connected.
\end{example}

Let $\phi \colon N \to M$ be a Real map between orientifolds. Then there is an induced homomorphism $\phi^* \colon H^0(M^\tau, \ZZ_2) 
\to H^0(N^\tau, \ZZ_2)$.  If $(P, Y)$ is a Real bundle gerbe on $M$, then 
$$
\phi^*\big(\sigma_2(P)\big) = \sigma_2\big(\phi^{-1}(P)\big) \ .
$$
In particular, if $m\in M^\tau$ is a fixed point, we can take $N = \{ m \}$ and
$\phi$ the inclusion to find $\sigma_2(P)(m) =  
\phi^*(\sigma_2(P)) = \sigma_2(\phi^{-1}(P))$.  In other words, we can
define the sign choice map by restricting $\sigma_2(P)$
to any fixed point $m \in M^\tau$, and declaring $\sigma_2(P)(m)$ to be $+1$ if the restricted Real bundle gerbe 
is Real trivial and $-1$ otherwise.  Thus we conclude

\begin{proposition}
The sign choice map \eqref{eq:sigma} for Real bundle gerbes $\cG$ coincides with the sign
choice of its Real Dixmier--Douady class $\DD_R(\cG)$ as defined in \eqref{eq:signchoice}.
\end{proposition}

Similarly to the cases of Real functions and equivariant line bundles, we  have

\begin{proposition} The composition 
 $$
 H^2(M, \cU(1)) \xrightarrow{Q \mapsto  Q \otimes \tau^{-1}(Q^*) } H^2(M; \ZZ_2, \cU^1) 
 \xrightarrow{ \ \sigma_2 \ } H^0(M^\tau, \ZZ_2)
 $$
 is equal to $1$.
\end{proposition}
\begin{proof}
If $(Q, X)$ is a bundle gerbe on $M$, define a Real bundle gerbe $(P, Y)$ by $Y = X \times_M \tau^{-1}(X)$
and $P  = Q \otimes \tau^{-1}(Q^*)$. 
A fibre point $y = (x, x')\in Y_m = X_m \times X_{\tau(m)}$ satisfies $\tau(x, x') = (x', x)$.  The fibre of $P = Q \otimes \tau^{-1}(Q^*)$ 
at $(y_1, y_2) = ((x_1, x'_1), (x_2, x'_2))$ is $Q_{(x_1, x_2)} \otimes Q^*_{(x'_1, x_2')}$.  If $\tau(m) = 
m$, we can pick $y=(x,x) \in Y_m$.  Any $f =
q\otimes q^* \in P_{(y, \tau(y))} = Q_{(x, x)} \otimes Q^*_{(x, x)}$ satisfies $\tau(f) = q^* \otimes q$ and $f\, \tau(f) = q( q^*) \otimes q(q^*) = 1 \in P_{(y, y)}$,
which maps to $1 \in U(1)$ as required.
 \end{proof}

Let $X=M/\tau$ denote the quotient space as in the orientifold
constructions of type~II string theory. By applying the Grothendieck
spectral sequence to the fibration $\pi\colon M \to X$, it is possible
to give a general characterisation of the kernel and image of the sign
choice map $\sigma_2$ in terms of the sheaf cohomology of the space of
``physical points'' $X$ of the orientifold $(M,\tau)$.  

\begin{theorem} 
\label{fiveterm} There exists a five-term exact sequence
$$
0\longrightarrow H^2(X, \pi^\tau_*\, \cU^1) \longrightarrow H^2(M;
\ZZ_2, \cU^1) \xrightarrow{ \ \sigma_2 \ } H^0(M^\tau,\ZZ_2)
\xrightarrow{\ \dd_3 \ } H^3(X,  \pi^\tau_*\, \cU^1)$$
where $\pi^\tau_*\, \cU^1$ is the invariant direct image sheaf.
\end{theorem}

\begin{proof} The  equivariant sheaf cohomology groups are the derived
  functors of the right exact functor $H^0(M; \ZZ_2, \cU^1) =
  {\mathit{\Gamma}}^\tau(\cU^1)$ of global invariant sections of the $\ZZ_2$-sheaf $\cU^1$:
$$
H^q(M; \ZZ_2, \cU^1) = (R^q\,{\mathit{\Gamma}}^\tau)(\cU^1) \ .
$$
Let ${\mathit{\Gamma}}$ denote the global sections functor on the quotient space $X$ and $\pi^\tau_*$ the invariant direct image functor. Then ${\mathit{\Gamma}}^\tau = {\mathit{\Gamma}} \circ \pi^\tau_*$ is the composition of two right exact functors, so the Grothendieck spectral sequence gives
$$
E^{p,q}_2 = (R^p\,{\mathit{\Gamma}})\,(R^q \, \pi^\tau_*)(\cU^1) =
H^p\big(X,(R^q\pi^\tau_*)\,(\cU^1) \big)
$$ 
 converging to 
$$
(R^{p+q}\,{\mathit{\Gamma}}^\tau)(\cU^1)= H^{p+q}(M; \ZZ_2, \cU^1) \ .
$$

The sheaf $(R^q\, \pi^\tau_*)(\cU^1)$ is trivial for
$q$ odd and equals the sheaf of constant $\ZZ_2$-valued functions
supported on $M^\tau$ for $q$ even. Namely if  $m\in M$ is not a fixed
point, then we can find a contractible open neighbourhood $m\in U$
such that $\tau(U)\cap U = \emptyset$, so $\pi(U)$ is an open
neighbourhood of $x=\pi(m) \in X$. The sheaf $(R^q\,
\pi^\tau_*)(\cU^1)$ is associated to the pre-sheaf $\big(V \subset
X\big) \mapsto H^q\big(\pi^{-1}(V);\ZZ_2, \cU^1|_{\pi^{-1}(V)} \big)$, and for $q>0$ we get  
$$ 
H^q\big(\pi^{-1}(\pi(U));\ZZ_2,  \cU^1 \big) = H^q\big(U \cup
\tau(U);\ZZ_2,  \cU^1 \big) =0  
$$
since $U$ is contractible. It follows that for $q>0$, the  sheaf $(R^q
\, \pi^\tau_*)(\cU^1)$ is supported on the fixed point set
$M^\tau$. Moreover, by choosing an invariant open neighbourhood around
each connected component of $M^\tau$ which equivariantly retracts onto
that O-plane, we conclude that the stalk of $(R^q\,
\pi^\tau_*)(\cU^1)$ at any point $x\in M^\tau$ is simply
$H^q(\pt;\ZZ_2,  \cU^1)$, which equals $0$ for $q$ odd and $\ZZ_2$ for
$q$ even. Since ${\rm Aut}(\ZZ_2) = 1$, there is also no monodromy.

Inserted into the Grothendieck spectral sequence, the result follows by the exact sequence of low degree terms since $E^{p,1}_2=0$, $E^{2,0}_2= H^2(X, \pi^\tau_*
\, \cU^1)$,  and the third differential $\dd_3$ maps $E^{0,2}_2=
H^0(M^\tau,\ZZ_2)$ to $E^{3,0}_2= H^3(X,  \pi^\tau_*\, \cU^1)$. 
\end{proof}

\begin{remark}
Theorem \ref{fiveterm} gives a genuine obstruction to possible sign choices: the
obstruction to surjectivity of $\sigma_2$ is the third differential
$\dd_3$ and the obstruction class is realised by an `orbifold
$2$-gerbe' on $X$ with band in the bundle of groups $\pi^\tau_*\,
\cU^1$, which lifts to a trivialisable Real bundle $2$-gerbe on
$M$ that corresponds physically to the flux of a flat three-form
Ramond--Ramond field; this is an obstruction only on orientifold
spacetimes which have O-planes of dimension a multiple of $4$~\cite{BGS,Keur}. If the obstruction map $\dd_3$ is trivial, then there is a short
exact sequence 
\begin{equation}
\label{eq:exactseq}
0\longrightarrow H^2(X, \pi^\tau_*\, \cU^1) \longrightarrow H^2(M;
\ZZ_2, \cU^1) \xrightarrow{ \ \sigma_2 \ } H^0(M^\tau,\ZZ_2)
\longrightarrow 0
\end{equation}
but it is not necessarily split.
We shall encounter some concrete examples of these facts in the following sections.
\end{remark}


\section{Sign choices for spherical orientifolds}
   \label{sec:constraints}     

Spheres provide important examples of orientifold compactifications of
string theory with background fluxes. In this section we will consider
the $n$-spheres $S^{n,1}$ for $n>1$ equipped with an involution with  fixed points, and
 apply the techniques from the previous sections to derive a general formula for the
geometric sign choices for a broad class of orientifolds.

  \subsection{The three-sphere $S^{3, 1}$}
  \label{sec:S13}

One of the most common examples of a string orientifold involves the three-sphere $S^{3, 1} \subset \RR^4$ equipped with the involution $(x, y, z, w) 
\mapsto (-x, -y, -z, w)$. This is equivalent to the Lie group $SU(2)$
equipped with the involution $\tau\colon g \mapsto g^{-1}$, so it has
two fixed points $\pm\, I$ corresponding to the north and south poles
of the three-sphere. The long exact sequence \eqref{eq:beginlongsequence}  yields 
\begin{equation*}
\xymatrix@R=2ex{ 
& 0 =  H^1(S^{3,1}, \cU(1)) \ \ar@{->}[r] & \ \ar@{->}[r] H^1(S^{3,1}; \ZZ_2, \cU^0) \ \ar@{->}[r] & \  H^2(S^{3,1}; \ZZ_2, \cU^1 ) \ \ar@{->}`r[d]`[lll]`[ddlll] `[ddll][ddll]&\\
&&&&\\
 & \   H^2(S^{3,1}, \cU(1)) \ \ar@{->}[r] & \  H^2(S^{3,1}; \ZZ_2,  \cU(1)) = 0 \  &    
}
\end{equation*}
where the vanishing of the last group follows from the spectral sequence for  
  Grothendieck's equivariant sheaf cohomology. Since $H^2(S^{3,1},
  \cU(1)) = H^3(S^3,\ZZ) =\ZZ $, we obtain 
   $$
 0  \longrightarrow  H^1(S^{3,1}; \ZZ_2, \cU^0) \xrightarrow{
   \ \partial_2 \ } H^2(S^{3,1}; \ZZ_2, \cU^1 ) \longrightarrow  \ZZ \longrightarrow  0 \ .
 $$ 
  The first group equals $\ZZ_2$ since it is the group
  of isomorphism classes of $\ZZ_2$-equivariant line bundles on $  S^{3,1}$. As any line
  bundle on $S^3$ is trivial, the only possible $\ZZ_2$-equivariant
  extensions are given by multiplication with $\tau = +1$ or $\tau =
  -1$ on the fibre $\CC$. This gives the
  Real Brauer group
 $$
  H^2(S^{3,1}; \ZZ_2, \cU^1 ) = \ZZ_2 \oplus \ZZ \ ,
 $$
where the complex Dixmier--Douady class of $(a,b)$ is $b$. This can
also be verified by using the Grothendieck spectral sequence.   
 
The group of sign choices is $H^0\big((S^{3,1})^\tau,\ZZ_2\big) =
\ZZ_2 \oplus \ZZ_2$. We proceed to show that the sign choice map 
   $$
 \sigma_2 \colon \ZZ_2 \oplus \ZZ \longrightarrow \ZZ_2 \oplus \ZZ_2
 $$
 sends $\ZZ_2$ to the diagonal, and that it maps the
 basic bundle gerbe generating $\ZZ$ to~$(-1, 1)$. 
  
For this, we will apply the construction of the sign choice map from
Example~\ref{ex:tautgerbe}, with $H$ given by $2\pi\ii$ times the normalized volume form on
$S^3$. Consider the group $SU(2)$ and take the basepoint $I$. We
regard $S^3$ as $\RR^3$ with the identity $I$ as $(0, 0, 0)$ and
$\tau(x, y, z) = (-x, -y, -z)$. Then the fixed point $-I$ is at
infinity. Let $p$ be the path from the origin along the positive $z$-axis to
infinity.  Applying $\tau$ then gives the path from the origin along the negative
$z$-axis to infinity. The surface $\Sigma\subset S^3$ is taken to be the set of
all $(x, 0, z)$ with $ x \geq 0$, and then $\tau(\Sigma) $ is the set
of all $(x, 0, z)$ with $x \leq 0$.  The union
$S=\Sigma\cup\tau(\Sigma)$ is the $(x,z)$-plane, which we take to bound
the ball $B_S$ given by the half-space $(x,y,z)$ with $y\geq0$. The
holonomy $\hol(\Sigma\cup\tau(\Sigma);H)$ is then computed by
integrating $H$ over half of $\RR^3$, or in other words over half of
$S^3$. The integral of $H$ over all of $S^3$ equals $2\pi\ii\,$, and hence the
sought integral
is $\pi\ii$ and the holonomy is $-1$. Thus the sign choice of
the basic bundle gerbe at the point $-I \in SU(2)$ is $-1$. At the
other pole of the sphere we can take all paths and $\Sigma$ constant,
giving the sign choice $+1$. 

We conclude that in the group $\ZZ_2\oplus\ZZ$ of Real bundle gerbes on
$SU(2)$, the basic bundle gerbe is $(0,1)$, while the Real bundle gerbe
coming from the connecting homomorphism, with sign choices $-1$ at
both fixed points, is $(-1,0)$. The geometric sign choices are
${\mathit\Sigma}_2(S^{3,1}) = \ZZ_2  \oplus \ZZ_2$, with the basic bundle
gerbe mapping to $(-1, 1)$ and the coboundary bundle gerbe mapping to
$(-1, -1)$.  Hence the map from $\ZZ_2 \oplus \ZZ$ to $\ZZ_2 \oplus
\ZZ_2$ is $(a, b) \mapsto  (a + [b], a - [b])$, where $b \mapsto [b]$
is the reduction map $\ZZ \to \ZZ_2$.  Thus the sign choice map is
surjective, but $\ZZ_2 \oplus \ZZ_2$ is not a direct summand in the Real
Brauer group. In particular, for all sign choices to arise one needs
to consider Real bundle gerbes with non-vanishing complex Dixmier--Douady
class.

Another way of seeing this is to use the Mayer--Vietoris
sequence. Write $S^{3,1} = U_+ \cup U_-$ as  the union of the upper
and lower hemispheres of $S^3$. These are equivariantly contractible and each
contribute $H^2(\pt; \ZZ_2, \cU^1 )=\ZZ_2$. The involution on the
equator $U_+\cap U_-$ acts freely and the quotient is the real projective plane $\RR
P^2$. Then 
$$
H^1(\RR P^2; \ZZ_2, \cU^1 )\simeq  H^2(\RR P^2, \ZZ_{-1}) \simeq H^0(\RR P^2, \ZZ) = \ZZ$$ 
where  $ \ZZ_{-1}$ is the orientation bundle on $\RR P^2$ and the
second isomorphism is Poincar\'e duality. Since $H^2(\RR P^2; \ZZ_2,
\cU^1 )= H^3(\RR P^2,\ZZ_{-1})=0$ for dimensional reasons, the Mayer--Vietoris sequence gives
$$
0\longrightarrow \ZZ \longrightarrow \ZZ_2\oplus \ZZ \xrightarrow{ \
  \sigma_2 \ }  \ZZ_2\oplus \ZZ_2\longrightarrow 0
$$
where the restriction map $\sigma_2$ sends  $(0,1)$  to $(-1,1)$. This is an instance of the exact sequence \eqref{eq:exactseq} where $\sigma_2$ is a non-split surjection.

\subsection{The $n$-sphere $S^{n,1}$}
\label{sec:n-1-sphere}

Let $S^{n,1}\subset \RR^{n+1}$ denote the $n$-sphere with $n\geq4$, equipped with the Real structure induced by the involution $\tau\colon (x_1, x_2,\dots, x_n,  x_{n+1}) 
\mapsto (-x_1, -x_2,\dots, -x_n, x_{n+1}) $.  The long exact sequence \eqref{eq:beginlongsequence}  yields
\[
0= H^1(S^{n,1}, \cU(1)) \longrightarrow  H^1(S^{n,1}; \ZZ_2, \cU^0) \longrightarrow
H^2(S^{n,1}; \ZZ_2, \cU^1 ) \longrightarrow  H^2(S^{n,1}, \cU(1)) = 0
\ .
\]
We have  $H^1(S^{n,1}; \ZZ_2, \cU^0)=\ZZ_2$ as the trivial line bundle
on $S^{n,1}$ admits two inequivalent Real structures. It follows that the Real Brauer group of $S^{n,1}$ is 
$H^2(S^{n,1}; \ZZ_2, \cU^1 )=\ZZ_2$. On the other hand, the involution
on $S^{n,1}$ has two fixed points, one at each pole of the sphere, so the group of sign choices is $H^0\big((S^{n,1})^\tau,\ZZ_2\big) = \ZZ_2^2 $. We conclude that $\sigma_2$ is not surjective in this case. 

More generally, consider a product of $k$ such Real spheres
$M=S^{n_1,1}\times \dots \times S^{n_k,1}$ with $n_i\geq 4$. Then
$H^2(M; \ZZ_2, \cU^1 )=\ZZ_2$ while $H^0(M^\tau,\ZZ_2) =
\ZZ_2^{ 2^k}$, so the sign choice map $\sigma_2$ is far from
being onto. However, in these cases it is more
natural to ask whether the sign choice map
$\sigma_{p+1}$ associated to $p$-gerbes with $p=n_1+\dots+n_k-2$ is
surjective. An explicit check is
hampered by our inability to geometrically define $p$-gerbes in
general. We consider below some explicit cases for $p=2$ 
where this can be addressed. The discussion would extend to a suitable
geometric realisation of $p$-gerbes for all $p>2$, wherein the
approach sketched below would work by 
inductively defining $\sigma_{p+1}$ in terms of $\sigma_{p}$.

\subsection{Sign choice for equivariant bundle 2-gerbes}\label{2-gerbe}

We briefly discuss the behaviour of the sign choice map $\sigma_3$. Write  $ \cG \To M$  for a bundle gerbe over a manifold $M$.
A {\em bundle $2$-gerbe}
$\mathscr{G}$ over $M$
consists of a surjective submersion $Y \to M$ together with a bundle
gerbe $\cG \To Y^{[2]}$ subject to some additional conditions detailed in \cite{Ste}.  
If $(M,\tau)$ is an orientifold, then we assume that $\mathscr{G}$ is equivariant in the strong sense that $\tau$ lifts to act on all the spaces by involutions.  The equivariant sheaf cohomology group $H^3(M; \ZZ_2, \cU^0) $ classifies $\ZZ_2$-equivariant bundle $2$-gerbes, up to  appropriate equivalence. 

 The connecting homomorphism $\partial_3$
sends Real bundle gerbes to $\ZZ_2$-equivariant bundle $2$-gerbes, and
$\sigma_3 \circ \partial_3 = \sigma_2$ by Proposition \ref{prop:sign-connecting}. Let $m \in M^\tau$, and let
$y \in Y_m$ be the corresponding fibre of $Y$ over $M$, then $\tau \colon \cG_{(y, \tau(y))} \to \cG_{(\tau(y), y)}$.  From the structure
of the bundle $2$-gerbe there is also an inversion map $\cG_{(\tau(y),
  y)} \to \cG_{(y, \tau(y))}$ which acts by conjugation. 
Thus composing the equivariant structure and the inverse gives a Real
structure $\cG_{(y, \tau(y))} \to \cG_{(y, \tau(y))} $, which has a
sign choice.  We declare this 
to be the sign choice $\sigma_3(\mathscr{G})(m)$ of the equivariant bundle
$2$-gerbe $\mathscr{G}$ over $M$ at the fixed point $m$. To justify
this definition, note that in the case of bundle gerbes the same construction would
define an equivariant action $P_{(y, \tau(y))} \to P_{(y, \tau(y))}$ which in the notation of Section \ref{sec:general} would be 
$f \mapsto \sigma(\cG)(m) \, \tau(f)^{-1} $, and we would take $\sigma(\cG)(m)$ as
our sign choice.  But then $\sigma(\cG)(m) = \< f\, \tau(f) \> =
\sigma_2(\cG)(m)$. 

\begin{example}[$\sigma_3$ is  onto]  We show that for $S^{4,1}$ the sign choice homomorphism
$$
\sigma_3 \colon H^3(S^{4,1}; \ZZ_2, \cU^0) \longrightarrow \ZZ_2 \oplus \ZZ_2
$$
 is surjective.  Since $S^{4,1}$ is $2$-connected and the
complex Dixmier--Douady class of any
Real bundle gerbe on $S^{4,1}$ vanishes, it follows from Proposition \ref{prop:constraint} below that $\sigma_2$ on such Real bundle gerbes is either $(1, 1)$ or $(-1, -1)$. Hence $\sigma_2$ on the corresponding 
equivariant bundle $2$-gerbes is the same.  To obtain other sign choices it  suffices to find an equivariant bundle $2$-gerbe
whose sign choice is $(1, -1)$ or $(-1, 1)$, as we can  get the
other sign choice by tensoring with the $(-1, -1)$ bundle $2$-gerbe.

We  
show that the bundle $2$-gerbe is the Chern--Simons bundle $2$-gerbe \cite{CarJohMurWan}
for the standard $SU(2)$-bundle $P\to S^4$. The total space of this bundle is given by 
$$
P = \big\{ (w^0, w^1) \in \HH^2 \ \big| \ w^0\, \bar w^0 + w^1\, \bar
w^1 = 1 \big\} = S^7 \ \subset \ \HH^2 = \RR^8 \ ,
$$
and the cosets of the right action by the unit quaternions $Sp(1) = SU(2) $ is $P_1\HH \simeq S^4$.
Define the orientifold action on $P$ by $\tau(w^0, w^1) = (-w^0, w^1)$, which
descends to the action on $S^{4,1}$. 
Over $[1, 0]$ the involution $\tau$ acts freely and over $[0, 1]$ it acts
trivially. The Chern--Simons bundle $2$-gerbe is defined as follows.
There is a map $\rho \colon P^{[2]} \to SU(2)$ defined by $p_1 = p_2\, \rho(p_1, p_2)$. 
Then the bundle gerbe over $P^{[2]}$ is the pullback of the basic
bundle gerbe $\cG \To SU(2)$, which we denote  by $\cH$. 
 The action of $\tau$ on $P$ satisfies $\rho \circ \tau = \rho$,
so the equivariant structure can be taken to be the identity 
$$
\cH_{(y_1, y_2)} = \cG_{\rho(y_1, y_2)} = \cG_{\rho(\tau(y_1),
  \tau(y_2))} =  \cH_{(\tau(y_1), \tau(y_2))} \ .
$$

There are two cases: 

\begin{itemize}

\item Let $y = (0, 1)$, then $\tau(y) = (0,1)$ so that $\rho(y,
\tau(y)) = I \in SU(2)$. Thus $\cH_{(y, \tau(y))} = \cG_I$.  The sign choice we want is that of this Real structure on $ \cG_{I}$, which 
we know from Section~\ref{sec:S13} equals $+1$.  

\item Let $y = (1, 0)$, then $\tau(y) = (-1,0)$ so that $\rho(y,
\tau(y)) = -I \in SU(2)$. The inversion map on the bundle gerbe is the
usual Real structure when identifying $SU(2) =
S^{3,1}$, so  the sign choice we want is that of this Real structure on $\cG_{-I}$, which 
we know from Section~\ref{sec:S13} equals $-1$.  

\end{itemize}

\end{example}

\begin{example}[$\sigma_3$ is not onto] Consider 
 the quadric $M = S^{2, 1} \times S^{2, 1}$. There
are four fixed points in this case, so the group of sign choices is $H^0(M^\tau,\ZZ_2)=\ZZ_2^{4}$.
An explicit calculation shows that all the geometric sign choices, including equivariant
bundle $2$-gerbes, are contained in the subgroup of sign choices which 
have an even number of $-1$ signs. Hence ${\mathit\Sigma}_3(M) \neq
H^0(M^\tau, \ZZ_2)$. 
\end{example}

\subsection{Constraints on geometric sign choices}
\label{sec:constraining}

We shall now derive a criterion that can often be used to constrain the possible
geometric sign choices on an orientifold $(M,\tau)$. 

\begin{definition} 
Let $m \neq m'$ be fixed points in $M^\tau$. A {\em Real $d$-span} of
$m$ and $m'$ is a Real map $\mu_{m,m'} \colon S^{d, 1} \to  M$ with 
$\mu_{m,m'}(1, 0, \dots, 0) = m$ and $\mu_{m,m'}(-1, 0, \dots, 0)=m'$. 
\end{definition}

The problem of finding a Real $d$-span may be formulated in terms of
certain constraints on $(M,\tau)$ as follows. Suppose that $M$ is $d-1$-connected with
$m$ and $m'$ isolated fixed points of $\tau$. Assume that it is
possible to choose a path $\gamma$
from $m$ to $m'$ which avoids the fixed points.  Then
$\tau(\gamma)$ is another path from $m$ to $m'$ that does not intersect $\gamma$
except at $m$ and $m'$, and $\gamma\cup\tau(\gamma)$ defines a map of $S^{1, 1}$
into $M$, which is a Real $1$-span since it commutes with the
involutions on $M$ and on $S^1$. Next if $d > 1$, then $M$ is
$1$-connected and suppose that a spanning surface $\Sigma$ for
$\gamma$ and $\tau(\gamma)$ can be chosen such that it avoids the fixed point set. Then $\tau(\Sigma)$ also has boundary
$\gamma\cup\tau(\gamma)$ and the union $\Sigma\cup \tau(\Sigma)$ determines a Real 
$2$-span. Continuing in this way, and as long as  the fixed
point set can be avoided, we obtain a Real $d$-span. 
 
\begin{proposition}
 \label{prop:constraint}
Let $\xi$ be an invariant $p$-gerbe over an orientifold $(M,\tau)$ with complex characteristic
class ${\rm C}_{p+2}(\xi)\in H^{p+2}(M,\ZZ)$ and sign 
 choice $\sigma_{p+1}(\xi) \in H^0( M^\tau , \ZZ_2)$. If $\mu_{m,m'}$
 is a Real $(p+2)$-span of $m,m' \in M^\tau$, then
\begin{equation}
\label{eq:signconstraint}
 \sigma_{p+1}(\xi)(m) = (-1)^{\< \mu_{m,m'}^*({\rm C}_{p+2}(\xi)),
   S^{p+2}\>} \ \sigma_{p+1}(\xi)(m'\,) \ .
 \end{equation}
 \end{proposition}
 \begin{proof}
Since the span is a Real map, it pulls back the invariant $p$-gerbe on $M$ to the sphere $S^{p+2, 1}$.
The sign of the pulled back $p$-gerbe at $(1,0,  \dots, 0)$ is $  \sigma_{p+1}(\xi)(m)$ 
and at $(-1,0,  \dots, 0)$ it equals $  \sigma_{p+1}(\xi)(m')$.
The complex characteristic class of the pulled back $p$-gerbe
is $\mu_{m,m'}^*({\rm C}_{p+2}(\xi))$.   The result follows by the calculation in Sections~\ref{sec:S13} and~\ref{sec:n-1-sphere}, which shows that an invariant $p$-gerbe on  $S^{p+2, 1}$ has 
signs that differ precisely by the periods of its complex characteristic class.
\end{proof}

The formula \eqref{eq:signconstraint} constrains the possible sign choices because once a sign
choice is made at one fixed point, all the other sign choices are
determined by the periods of the class of the $p$-gerbe. This is
independent of the Real or equivariant structure on the $p$-gerbe and
depends solely on the periods of ${\rm C}_{p+2}(\xi)$. In
particular, if the $p$-gerbe $\xi$ has vanishing complex characteristic
class ${\rm C}_{p+2}(\xi)=0$, then the sign choice map
$\sigma_{p+1}(\xi)$ is constant and all sign choices are the same.
 
 
\section{Sign choices for toroidal orientifolds}
  \label{prop:nonsimplyconn}

Tori constitute important examples of flat orientifold
compactifications of string theory.
 In this section we consider examples of multiply connected
 orientifolds involving tori. The allowed sign choices in these cases have been studied from different perspectives in~\cite{deBDHKMMS,GaoHor,Ros2013,DMR2}.

\subsection{The circle $S^{1,1}$}

We denote by $S^{1, 1}$  the circle $S^1=U(1)$ with the Real structure
$\tau(u) = \bar u$ and fixed points $\pm\, 1$. We 
 can identify $H^0\big( (S^{1,1})^\tau, \ZZ_2\big) $ with $\ZZ_2
 \oplus \ZZ_2$ by the map $f \mapsto (f(1), f(-1))$. There are two
types of string theories to which this pertains to: the orientifold type~IIA string theory
compactified on $S^{1,1}$ is the type~IA theory which under T-duality
maps to type~I string theory on a circle, or equivalently the orientifold
type~IIB 
theory on $S^{0,2}$, whereas the type~$\widetilde{\rm
  IA}$ theory is the T-dual of the type~IIB orientifold on
$S^{2,0}$~\cite{Hori,OlsSza}. The two theories are
distinguished by the fact that in the former case the two O-planes
have the same sign while in the latter case they have opposite
signs. Both theories arise within our present formalism if all
sign choices are possible, which is guaranteed by
\begin{proposition}
The sign choice map 
 \begin{equation}
 \label{eq:circle}
 \sigma_2 \colon H^2(S^{1,1}; \ZZ_2, \cU^1) \longrightarrow
 H^0\big((S^{1,1})^\tau, \ZZ_2 \big) = \ZZ_2^2 
 \end{equation}
 is an isomorphism. 
  \end{proposition}
  \begin{proof}
 The long exact 
 sequence  \eqref{eq:beginlongsequence} in this case degenerates to 
 \begin{equation}
 \label{eq:circle-connecting}
 0 \longrightarrow H^1(S^{1,1}; \ZZ_2, \cU^0)  \longrightarrow
 H^2(S^{1,1}; \ZZ_2, \cU^1 ) \longrightarrow 0 \ ,
 \end{equation}
 so we need to understand the space of equivariant line
 bundles over $S^1$. Every  line bundle over $S^1$ is trivial and the
 generic $\ZZ_2$-action on the trivial bundle  is $g\, \tau$, where
 $g:S^{1}\to U(1)$ and $\tau$ is the trivial action on 
 the trivial bundle.
  To compute the number of inequivalent $\ZZ_2$-structures on the trivial bundle, we need to calculate the quotient
 in Proposition \ref{prop:Z2-structures}.  We proceed by showing that the sequence
 \begin{equation}
 \label{eq:circle-exact}
 H^0(S^{1}, \cU(1)) \xrightarrow{f \mapsto f\, (\bar f \circ \tau) }  H^0(S^{1,1}; \ZZ_2, \cU^1) 
 \xrightarrow{ \ \sigma_0 \ } H^0\big((S^{1,1})^\tau, \ZZ_2\big) \longrightarrow 0
 \end{equation}
 is exact, from which it follows that 
 $$
 \sigma_1 \colon H^1(S^{1,1}; \ZZ_2, \cU^0) \longrightarrow H^0\big((S^{1,1})^\tau, \ZZ_2\big)
 $$
 is an isomorphism. The result then follows by  the connecting isomorphism \eqref{eq:circle-connecting}.
   
To establish the exactness of \eqref{eq:circle-exact},  we show that
for any $g \in H^0(S^{1,1}; \ZZ_2, \cU^1) $ with $g(1) = 1$ and $g(-1)
= 1$, there exists a function $f \colon S^1 \to U(1)$ with 
$ g(u)  = f(u) \, \overline{ f(\bar u) }$.
Identify $S^1 = \RR / \ZZ$ and lift $g$ to $\hat g \colon \RR \to
\RR$. Then $\hat g(t + 1) = \hat g(t) + n$ for some $n \in \ZZ$ which is the winding number of $g$. 
  Because $g(1) = 1$   we can assume that $\hat g (0 ) = 0$, and then by $g(-1) = 1$ it follows that  $\hat g(\frac{1}{2}) = k$ for some $k \in \ZZ$.  Also  because  $\overline{g (\bar u) } = g(u) $ for all $u \in S^1$, we have 
$$
- \hat g( 1 - t) = \hat g(t) + m 
$$
for some $m \in \ZZ$. Taking $t=\frac12$ and $t=0$ then shows that
$m=-2k$ and $n=2k$.
Define $\hat f(t ) = \frac12\, \hat g(t)$, then 
$$
\hat f(t +1) = \hat f(t)  + k 
$$
and 
\begin{align*}
\hat f (t) - \hat f(1 - t) = \tfrac12\,\hat g(t) - \tfrac12\,\hat g(1
  - t) =  \hat g(t) -  k \ . 
\end{align*} 
We conclude that $\hat f$ descends to  a well-defined function $f
\colon S^1 \to U(1)$   satisfying $ f(u) \, \overline{ f(\bar u) } = g(u)$ as required.
  \end{proof}
 
\subsection{The two-torus $S^{1,1} \times S^{1,1}$}
\label{sec:2torus}

The case of a two-torus is analogous  but more complicated since now there are non-trivial line bundles on $T^2=S^{1,1} \times S^{1,1}$.  These are classified by $H^2(T^2, \ZZ) = \ZZ$.  A line bundle whose Chern class is the generator can be described as follows \cite{GaoHor}. 
 If $\boldsymbol{ x},\boldsymbol{ y} \in \RR^2$ let 
 $$
 c({\boldsymbol x, \boldsymbol y}) = \exp\big( 2 \pi \ii x^1\, ( x^2 -
 y^2) \big) \ .
 $$
   Define an equivalence relation  on $\RR^2 \times U(1)$ by $(\boldsymbol{ x}, u) \sim (\boldsymbol{ y}, v)$ if $\boldsymbol{ x} - \boldsymbol{  y} \in \ZZ^2$ and $v =  c(\boldsymbol{ x} , \boldsymbol{ y} )\, u$. The sought $U(1)$-bundle $P$ is defined as the quotient by this equivalence relation, with the projection $\pi([\boldsymbol{ x} , u]) = [\boldsymbol{ x} ]$  where $[\boldsymbol{ x}]$ denotes the orbit in $S^1 \times S^1 = \RR^2 / \ZZ^2$.  We lift the Real structure on $S^{1,1} \times S^{1,1}$ to $P$ by defining $\tau([\boldsymbol{ x}, u]) = [-\boldsymbol{ x}, u]$. Checking that this is well-defined  reduces to observing that $c(-\boldsymbol{ x}, -\boldsymbol{ y}) = c(\boldsymbol{ x}, \boldsymbol{ y})$.  
 The fixed points of $\tau$ on $T^2$ are $(0, 0)$, $(\frac12, 0)$, $(0, \frac12)$ and $(\frac12, \frac12)$.  On the corresponding fibres the Real 
 structure acts by 
 $$
 \tau\big([\boldsymbol{ x}, u]\big) = [-\boldsymbol{ x}, u] =
 \big[\boldsymbol{ x}, \exp( 4 \pi \ii x^1\, x^2)\, u \big] \ ,
 $$
so  $\tau$ acts by the identity on the four fixed points. Hence $\sigma_1(P) = (1, 1, 1, -1)$
 where we denote a sign choice by its values on the four fixed points as listed above. 
 In other words, every odd power of $P$ admits a $\ZZ_2$-structure with this same sign choice and even powers have 
 trivial sign choice.
 
Next we need to study the quotient in Proposition \ref{prop:Z2-structures} and consider the sequence 
 $$
  H^0(T^2, \cU(1)) \xrightarrow{f \mapsto f\, (\bar f \circ \tau) }  H^0(T^2; \ZZ_2, \cU^1) 
 \xrightarrow{ \ \sigma_0 \ } H^0\big((T^2)^\tau, \ZZ_2\big) = \ZZ_2^4 \ .
 $$
 Let $E \subset H^0\big((T^2)^\tau, \ZZ_2\big) $ denote the subgroup
 of sign choices whose product is $+1$, that is,
 there are an even number of $-1$'s.  Then $E \simeq \ZZ_2^{3}$ as any element is determined by its value
 on just three of the four fixed points. We  have
 
 \begin{proposition}
\label{prop:imageE}
 The image of $\sigma_0$ is $E$.
 \end{proposition}
 \begin{proof}
Consider the functions $g\colon T^2 \to U(1)$, $(u,v)\mapsto g(u, v)$ defined by the values $\pm\, 1$, $\pm\, u$, $\pm\, v$ and $\pm\, u\,v$.  These give all the sign choices with even numbers of negative signs. 
 Assume now that there is a function $g$ with $g(1, 1) = 1$ and $g(1, -1) = g(-1, 1) = -1$. We show that it must have $g(-1, -1) = 1$. 
 We may lift $g$ to $\hat g \colon \RR^2 \to \RR$ in such a way that $\hat g(0, 0) = 0$. Moreover we must have $\hat g(x+1, y) = \hat g(x, y) + m$, 
 $\hat g(x, y+1) = \hat g(x, y) + n$ and $-\hat g(-x, -y) = \hat g(x,
 y) + k$ for some integers $m$, $n$ and $k$. Substituting $\hat g(0, 0) = 0$ gives $k = 0$.
 In addition we  have $\hat g(\frac12, 0) = p + \frac12$ and $\hat
 g(0, \frac12) = q + \frac12$ for some $p, q \in \ZZ$, so $\hat g(\frac12, \frac12) = - \hat g(-\frac12, -\frac12)$ and $\hat g(\frac12, \frac12) = \hat g(-\frac12, -\frac12) + m + n$.  Thus $\hat g(\frac12, \frac12) = \frac12\,(m+n)$. 
 Applying a similar argument
 gives $\hat g(\frac12, 0) = \frac m2 = p+\frac12$ and $\hat g(0,
 \frac12) = \frac n2 = q+\frac12$. Hence $\hat g(\frac12, \frac12) = p+q + 1 \in \ZZ$ giving the required result. 
Now any sign choice with an odd number of negative signs is a sign choice with an even number of negative signs multiplied by $(1, -1, -1, 1)$, so none
 of them can occur in the image of $\sigma_0$.
  \end{proof}

 \begin{proposition} 
 If $\sigma_0(g) = (1, 1, 1, 1) $, then there exists a smooth function $f \colon T^2 \to U(1)$ such that $g(u, v) = f(u, v)\, \overline{ f(\bar u, \bar v)}$.
 \end{proposition}
 \begin{proof}
 As in the proof of Proposition~\ref{prop:imageE} we lift $g$ to $\hat g \colon \RR^2 \to \RR$ such that $\hat g(x+1, y) = \hat g(x, y) + m$
 and $\hat g(x, y+1) = \hat g(x, y) + n$. We may assume that $\hat g(0, 0) = 0$ and thus $- \hat g(-x, -y) = \hat g(x, y)$. 
 Moreover $\hat g(\frac12, 0) = p$, $\hat g(0, \frac12) = q$ and $\hat g(\frac12, \frac12) = r$ where $p, q, r \in \ZZ$.  Hence $\hat g(\frac12, 0) = \hat g(-\frac12, 0) + m 
 $ and $ p = \hat g(\frac12, 0) = - \hat g(-\frac12, 0)$ so that $m = 2p$. Similarly $n = 2q$. Define
 $$
\hat f(x, y) = \tfrac12\, \hat g(x, y) 
 $$
 and notice that $\hat f(x+1, y) = \hat f(x, y) + p$ and $\hat f(x, y+1) = \hat f(x, y) + q$, so that $\hat f$ descends to a well-defined function $f \colon T^2 \to U(1)$. 
 Moreover
 $$
 \hat f(x, y) - \hat f(-x, -y) = \tfrac12\, \hat g(x, y) - \tfrac12\,
 \hat g(-x, -y) = \hat g(x, y) \ ,
 $$
 so that $g(u, v) = f(u, v)\, \overline{ f(\bar u, \bar v)}$ as required.
 \end{proof}
 
It follows that the sequence
  \begin{equation}
 \label{eq:torus-long}
  H^1(T^2, \cU(1)) \xrightarrow{P\mapsto P  \otimes \tau^{-1}(P) }
  H^1(T^2; \ZZ_2, \cU^0) \xrightarrow{ \ \partial_2 \ } H^2(T^2; \ZZ_2, \cU^1)  \longrightarrow 0
 \end{equation}
 is exact since $T^2$ is two-dimensional.   Moreover we have
 $$
 H^1(T^2; \ZZ_2, \cU^0) = \ZZ \oplus \ZZ_2^{3}
 $$
 as line bundles on $T^2$ with arbitrary Chern class admit any choice of $\ZZ_2$-action.  From the calculation above
 of the sign choice for the generator in $H^2(T^2, \ZZ)$, we conclude that  $\sigma_1$ maps 
 $\ZZ \oplus \ZZ_2^{3}$ to  $\ZZ_2 \oplus \ZZ_2^{3} = \ZZ_2^{4}$ surjectively. This agrees with the results of~\cite{GaoHor}, which
 we will explain in more generality in Section~\ref{sec:KR-theory}.
  
  The image of the map $P\mapsto P  \otimes \tau^{-1}(P)$ 
consists of the line bundles with even Chern class, so the Real Brauer group is 
 $$
 H^2(T^2; \ZZ_2, \cU^1 )  = \ZZ \oplus \ZZ_2^{3}
 $$
 and the action of $\sigma_2$ is again projection onto
 $\ZZ_2^{4}$. In summary, we have
 
  \begin{proposition}
The Real Brauer group of the  two-torus $T^2=S^{1, 1} \times S^{1, 1}$ is
  $$
 H^2(T^2; \ZZ_2, \cU^1 )  = \ZZ \oplus \ZZ_2^{3}
 $$
 and  the sign choice map  $\sigma_2$ is the projection onto
 $H^0\big((T^2)^\tau, \ZZ_2\big) =
 \ZZ_2^{4}$ induced by  reduction modulo 2 
 on the first factor. 
 \end{proposition}
 
\subsection{The three-torus $S^{1,1} \times S^{1,1} \times S^{1,1}$}

 With the conjugation action on each circle factor, the involution on the
 three-dimensional torus $T^3$ has eight fixed points and hence there are 
 $2^8$ possible sign choices. Since $T^3 = \RR^3 / \ZZ^3 $, the fixed points are the image of the eight points
 $(\rho_1, \rho_2, \rho_2)$ where $\rho_i\in\{0,\frac12\}$.   We can regard the fixed points on the three-torus as forming a cube.  Let $E\subset H^0\big( (T^3)^\tau, \ZZ_2\big)$ be the subgroup of all sign choices with the property that their product around each square of the cube equals $1$. 
Then it is straightforward to see that $E$ contains $16$ elements. We
now show that the sequence
  $$
  H^0(T^3, \cU(1)) \xrightarrow{f \mapsto f\, (\bar f \circ \tau) }  H^0(T^3; \ZZ_2, \cU^1) 
 \xrightarrow{ \ \sigma_0 \ } E \longrightarrow 0
 $$
is well-defined, that is, $\sigma_0$ takes values in $E$, and that it
is exact.   Consider $f \colon T^3 \to U(1)$ satisfying $f(z) \, \overline{f(\bar z)} = 1$ for all  $z \in T^3$, and lift $f$ to the 
 universal cover $\hat f \colon \RR^3 \to \RR$.  If $\boldsymbol{ e}_1, \boldsymbol{ e}_2, \boldsymbol{ e}_3$ is the standard basis of $\RR^3$, then 
 $\hat f(\boldsymbol{ x} + \boldsymbol{ e}_i) = \hat f(\boldsymbol{
   x}) + n_i$ for $n_i \in \ZZ$ and $\hat f(\boldsymbol{ x}) + \hat
 f(-\boldsymbol{ x}) = k \in \ZZ$, for all
 $\boldsymbol{x}\in\RR^3$. If $\epsilon_i
 \in \{ 0, 1 \}$ for $i = 1, 2, 3$, then it follows that
 \begin{equation}
 \label{eq:3torus-sign}
 \hat f\Big( \frac{1}{2} \, \mbox{$\sum\limits_{i=1}^3$}\, \epsilon_i\,
 \boldsymbol{ e}_i  \Big) = \frac{1}{2}\, \bigg(k +
 \sum\limits_{i=1}^3 \, n_i\, \epsilon_i\bigg) \ .
 \end{equation}
Let $i, j, k$ be distinct and $\epsilon \in \{ 0, 1\}$, then 
 adding up around a square gives
$$
 \hat f\big(\tfrac{1}{2}\, \epsilon\, \boldsymbol{ e}_k\big) {+} \hat f\big( \tfrac{1}{2}\, \epsilon\, \boldsymbol{ e}_k {+} \tfrac{1}{2}\, \boldsymbol{ e}_i \big) {+} \hat f\big( \tfrac{1}{2}\, \epsilon\, \boldsymbol{ e}_k + \tfrac{1}{2}\, e_j\big) {+}
  \hat f\big( \tfrac{1}{2}\, \epsilon \, \boldsymbol{ e}_k + \tfrac{1}{2}\, \boldsymbol{ e}_i + \tfrac{1}{2}\, \boldsymbol{ e}_j\big) 
  = 2\, k {+} 2\, \epsilon\, n_k {+} n_i {+} n_j \ \in \ \ZZ
$$
and  so multiplying the values of $f$ around the corresponding square of fixed points gives $1$. 
  Hence the image of $\sigma_0$ is contained in $E$.
  
 Next we show that $\sigma_0$ is surjective.  For each $S^1$ factor there are two Real functions $1$ and $z$, combining these makes
 eight Real functions and applying $\pm\, 1$ makes $16$.  They are all distinct so $\sigma_0$ is surjective. 
 
 Finally we prove exactness in the centre.  That is, if $\sigma_0(f) =
 1$ we  show that there is a function $g : T^3 \to U(1)$ such that
 $f(z) = g(z)\, \overline{g(\bar z)}$.  Lift $f$ to $\hat f \colon \RR^3 \to \RR$ as above.  Then from \eqref{eq:3torus-sign}
 we deduce that 
 $$
 \frac{1}{2}\,\bigg(k + \sum_{i=1}^3\, n_i\, \epsilon_i\bigg)
 $$
 is always an integer and hence that $k, n_1, n_2, n_3 \in 2\, \ZZ$.
 The function
 $$
 \hat g(\boldsymbol{ x}) = \tfrac12\, \hat f(\boldsymbol{ x})
 $$
 satisfies $\hat g(\boldsymbol{ x} + \boldsymbol{ e}_i) - \hat
 g(\boldsymbol{ x}) \in \ZZ$ and so descends to a 
 well-defined function $g \colon T^3 \to U(1)$. Moreover we obtain 
 \begin{align*}
\hat g(\boldsymbol{ x}) - \hat g(-\boldsymbol{ x}) = \tfrac{1}{2}\,
   \hat f(\boldsymbol{ x}) - \tfrac{1}{2}\, \hat f(-\boldsymbol{ x}) =
   \hat f(\boldsymbol{ x}) - \tfrac{1}{2} \, k
              \end{align*}
              as required.  This gives the claimed exactness.   This
              also shows that $E \simeq \ZZ_2^4$.
  
  Consider now the exact sequence
  \begin{equation*}
 H^1(T^3, \cU(1)) \xrightarrow{P\mapsto P  \otimes \tau^{-1}(P) }
  H^1(T^3; \ZZ_2, \cU^0) \xrightarrow{ \ \partial_2 \ }  H^2(T^3;
  \ZZ_2, \cU^1 ) \longrightarrow H^2(T^3, \cU(1)) \ .
 \end{equation*}
It follows from the K\"unneth theorem that
$
  H^1(T^3, \cU(1)) = \ZZ^3,
$
  and each of the generating line bundles is a pullback via a projection $T^3 \to T^2$ of the standard bundle
  constructed in Section~\ref{sec:2torus}, so each admits a $\ZZ_2$-action.  Hence 
  $$
   H^1(T^3; \ZZ_2, \cU^0) \simeq \ZZ^3 \oplus E \simeq \ZZ^3 \oplus
   \ZZ_2^4 \ ,
   $$
with the sign choice map $\sigma_1$ being the projection of $\ZZ^3 \oplus \ZZ_2^4
$ onto $\ZZ_2^7$.  The image consists of all sign choices whose
product over the eight fixed points is equal to $1$. 
The map $P\mapsto P  \otimes \tau^{-1}(P)$ is  multiplication by $2$ on the $\ZZ^3$ factor. 

Finally we consider Real bundle gerbes on $T^3$.  Let $Y = \RR^3 \to
T^3$ be the universal cover and define $g \colon Y^{[3]} \to U(1)$
by 
$$
g(\boldsymbol{ x}, \boldsymbol{ y}, \boldsymbol{ z}) = \exp\big(2 \pi
\ii x^1\, (x^2 - y^2)\, (y^3-z^3)  \big) \ .
$$
  It is straightforward to show that $\delta(g) \colon Y^{[4]} \to U(1) $
 equals $1$.  This gives a bundle gerbe $P = Y^{[2]} \times U(1)$ with
 multiplication 
$$
\big((\boldsymbol{ x}, \boldsymbol{ y}), u \big) \,
\big((\boldsymbol{ y}, \boldsymbol{ z}), v \big) = \big((\boldsymbol{
  x}, \boldsymbol{ z}) \,,\, g(\boldsymbol{ x}, \boldsymbol{ y}, \boldsymbol{ z})\, u\, v \big) \ .
$$
There is a lift of the Real structure on $T^3$ to $Y$ by
$\tau(\boldsymbol{ x}) = -\boldsymbol{ x}$. Since $g(-\boldsymbol{
  x}, -\boldsymbol{ y}, -\boldsymbol{ z}) = \overline{ g(\boldsymbol{
    x}, \boldsymbol{ y}, \boldsymbol{ z}) },$ we obtain a Real structure on $P$ given by  $\tau( (\boldsymbol{ x}, \boldsymbol{ y}), u) = ( (-\boldsymbol{ x}, -\boldsymbol{ y}), \bar u)$. 
One easily verifies that this bundle gerbe has complex Dixmier--Douady class generating $H^2(T^3, \cU(1)) = \ZZ$. Hence 
$$
H^2(T^3; \ZZ_2, \cU^1 ) = \ZZ_2^3 \oplus E \oplus \ZZ \ .
$$
We  calculate the sign of the Real bundle gerbe  following the
procedure of Section~\ref{sec:2torus}. 
We 
take $\boldsymbol{ x} \in Y$ so that $(\boldsymbol{ x}, -\boldsymbol{ x}) \in Y^{[2]}$ and pick $p = (\boldsymbol{ x}, -\boldsymbol{ x}, 1) \in P_{(\boldsymbol{ x}, -\boldsymbol{ x})}$.  Then
$\tau(p) = ( -\boldsymbol{ x}, \boldsymbol{ x}, 1)$ and 
$$
p\, \tau(p) = \big((\boldsymbol{ x}, \boldsymbol{ x}), g(\boldsymbol{
  x}, -\boldsymbol{ x}, \boldsymbol{ x}) \big) = \big( (\boldsymbol{
  x}, \boldsymbol{ x}), \exp(- 8 \pi \ii x^1 \, x^2 \, x^3 ) \big) \ . 
$$
Thus the sign choice is $1$ except at $(\frac12, \frac12, \frac12)$ where it equals $-1$.

 \begin{proposition}
The Real Brauer group of the three-torus $T^3 = S^{1, 1} \times S^{1, 1} \times S^{1, 1}$ is
  $$
 H^2(T^3; \ZZ_2, \cU^1 )  = \ZZ^7_2  \oplus \ZZ \ ,
 $$
 and  the sign choice map  $\sigma_2$  is the projection onto
 $H^0\big((T^3)^\tau,\ZZ_2\big) = \ZZ_2^8$
 induced by  reduction modulo 2 
 on the last factor.
 \end{proposition}


\section{$\KR$-theory with sign choices}
  \label{sec:KR-theory}

In this final section we make contact with the prescription of~\cite{GaoHor}
for computing the sign choices in orientifold backgrounds of type~II
string theory whose complex Dixmier--Douady
class is trivial, that is, when $H=\dd B$ globally. This enables a general $\KR$-theory
classification of D-brane charges in the cases when the sign choice
map is not constant over the orientifold planes. 
  
\subsection{The Gao--Hori construction}
\label{sec:GaoHori}

From explicit worldsheet considerations, Gao and Hori in~\cite{GaoHor}
gave the relevant geometric ingredients required of any type~II string
orientifold construction with topologically trivial $H$-flux, which we briefly
review and then translate into the language of the present paper. Let
$(M,\tau)$ be an orientifold. Then the data of~\cite{GaoHor} consists
of a quadruple $(B,L,\alpha,c)$ where:
\begin{itemize}
\item \ $B\in\Omega^2(M)$ is a globally defined $B$-field on $M$.
\item \ $L$ is a line bundle over $M$ called the twist bundle.
\item \ $\alpha$ is a connection on $L$ called the twist connection.
\item \ $c$ is a global section of $\tau^{-1}(L)\otimes L^*$ called the
  crosscap section.
\end{itemize}
The line bundle with connection $(L,\alpha)$ originates from requiring
invariance of the $B$-field amplitude (see
Section~\ref{sec:Bfieldampl}) under the parity transform
$\hat\phi\mapsto\tau\circ\hat\phi\circ\Omega$ for all string
worldsheets $\hat\phi:\widehat{\Sigma}\to M$, which for a globally defined $B$-field implies that the two-form $B+\tau^*(B)$ has integral periods, and so represents a class in $H^2(M,\ZZ)$ or equivalently the curvature of a connection on a line bundle over $M$. Invariance of D-brane Chan--Paton factors under the parity transform then requires that the twist bundle $L \to M$ be equivariant and that the twist connection $\alpha$
on $L$ is a $\ZZ_2$-invariant connection. The crosscap section $c$ is a $\ZZ_2$-structure
on the twist bundle, and altogether it follows that the quadruple $(B,L,\alpha,c)$ is subjected to the following constraints:
\begin{itemize}
\item[(C1)] \ $\dd\alpha = B+\tau^*(B)$.
\item[(C2)] \ The connection $\tau^*(\alpha)-\alpha$ of
  $\tau^{-1}(L)\otimes L^*$ is flat and has trivial holonomy.
\item[(C3)] \ $c$ is a parallel section with respect to the connection
  $\tau^*(\alpha)-\alpha$, that is, $\dd c+\big(\tau^*(\alpha)-\alpha\big)\,c=0$, such that $c\,\tau^*(c)=1$.
\end{itemize}
The condition (C3) implies that $\dd c=0$ and $c^2=1$ on the fixed point set $M^\tau$, and so the crosscap section determines a sign choice $[c]\in H^0(M^\tau,\ZZ_2)$, which agrees with our definition of sign choice from Section~\ref{sec:signbundle}. When the twist $(L,\alpha)$ is trivial, $c$ is constant on $M$ and all O-planes have the same sign choice. For non-trivial twists, the values of $c$ can differ from one O-plane to another.

We will now show that the twist connection $\alpha$
on $L$ gives rise to a Real
bundle gerbe connection on a trivial bundle gerbe $(Q, Y)
= \partial_2(L)$, and the $B$-field gives rise to a Real curving for this connection.
For this, we slightly unravel the coboundary map from the definition
in~\cite{HekMurSza}. Given the equivariant bundle $L \to M$, so that $\tau^{-1}(L) = L$, let $Y = M \times \ZZ_2$ with the involution $\tau \colon (m , x)  \mapsto (\tau(m), x + 1)$ covering $\tau$ on $M$. Let $\pi_M \colon Y \to M$ be the projection. Then $Y$ is two copies of $M$ labelled by $0$ and $1$ so that any bundle $Q \to Y$ is a pair of bundles $(Q_0, Q_1)$
on $M \times \{0\}$ and $M \times \{1\}$, respectively. For such a
bundle $Q$ we have $\tau^{-1}(Q) = (\tau^{-1}(Q_1), \tau^{-1}(Q_0))$, and if $L \to
M$ is a line bundle then 
$\pi_M^{-1}(L) = (L, L) \to Y$.  
Following \cite{HekMurSza} we define the Real bundle gerbe 
over $M$ which is the image of $L$ under the connecting homomorphism
as $Q \to Y^{[2]} = M \times \ZZ_2\times \ZZ_2$ with $Q
= \delta(U(1), L)$, where here $U(1)$ denotes the trivial bundle
over $M$.
Explicitly, $Q$ is the disjoint union of four bundles over $M$ labelled by 
elements of $\ZZ_2\times\ZZ_2$ given by 
$$
Q_{(0,0)} = U(1) \ , \quad Q_{(0, 1)} = L \ , \quad Q_{(1, 0)} = L^*
\qquad\text{and}\qquad Q_{(1, 1)} = U(1) \ .
$$
The bundle gerbe multiplication is made up of the obvious contractions and the Real structure 
follows from 
\begin{align*}
\tau^{-1}(Q^*)_{(0, 0)} &= \tau^{-1}(Q^*_{(1, 1)}) = \tau^{-1}(U(1)^*)
                          = U(1) = Q_{(0,0)} \ , \\[4pt]
\tau^{-1}(Q^*)_{(1, 0)} &= \tau^{-1}(Q^*_{(0, 1)}) =\tau^{-1}(  L^* )
                          = L^* = Q_{(1, 0)} \ , \\[4pt]
\tau^{-1}(Q^*)_{(0, 1)} &= \tau^{-1}(Q^*_{(1, 0)}) =\tau^{-1}(  L) = L
                          = Q_{(0, 1)} \ , \\[4pt]
\tau^{-1}(Q^*)_{(1, 1)} &= \tau^{-1}(Q^*_{(0, 0)}) = \tau^{-1}(U(1)^*)
                          = U(1) = Q_{(1,1)} \ .
\end{align*}

We can pick a connection $\alpha$ on $L$ which is invariant by picking any
connection and averaging it.  Being invariant it automatically
satisfies the conditions (C2) and (C3) above. 
Then the bundle $Q$ inherits a connection which is a bundle gerbe connection and also Real. 
If we call this bundle gerbe connection $A^Q$, then on the four components of $Y^{[2]}$ we have
$$
A^Q_{(0, 0)} = 0 \ , \quad A^Q_{(0,1)} = \alpha \ , \quad A^Q_{(1,0)}
= -\alpha \qquad \mbox{and} \qquad A^Q_{(1, 1)} = 0 \ ,
$$
so its curvature $F^Q=\dd A^Q$ has components 
$$
F^Q_{(0, 0)} = 0 \ , \quad F^Q_{(0,1)} = \dd\alpha \ , \quad
F^Q_{(1,0)} = -\dd\alpha \qquad \mbox{and} \qquad F^Q_{(1, 1)} = 0 \ .
$$
We can take the curving on $Y$ to have components $f^Q = (-\dd\alpha, \dd\alpha)$.  The $B$-field 
is a two-form on $M$ satisfying $B + \tau^{*}(B) = \dd\alpha $. If we define a two-form $f$ on $Y$
by $ f = (f_0 , f_1) = (B, -\tau^*(B))$, then to be a curving we would need $F^Q_{(i, j)} = f_j - f_i$,
which is indeed satisfied. 

\subsection{Twisted $K$-theory}

Gao and Hori  use the orientifold data $(B,L,\alpha,c)$ in \cite{GaoHor} to introduce a new variant of $\KR$-theory groups $\KR(M;c)$ accommodating non-constant sign choices, and conjecture a certain Fredholm module formulation for this $K$-theory.
We will show that the $K$-theory of vector bundles \emph{twisted
  by equivariant line bundles} introduced in \cite{GaoHor} is simply
$\KR$-theory twisted by a certain Real bundle gerbe. From our general twisted $\KR$-theory construction from~\cite{HekMurSza}, this implies in
particular that the twisted $K$-theory construction of \cite{GaoHor}
is a generalised cohomology theory, and the Fredholm module description follows immediately.

First we note that there is a natural notion of modules associated to each invariant $p$-gerbe for $p=-1,0,1$. These modules form a semi-group, whose Grothendieck group completion determines a $K$-theory group. The respective notions of module are defined as follows.
\begin{definition}\label{def:modules}
Let $(M,\tau)$ be an orientifold.
\begin{itemize}
\item[(a)]
If $c \colon M \to U(1)$ is a Real function, then a \emph{$c$-module} is a vector bundle $E \to M$ with a 
 conjugate linear map $\tau_E:E \to E$ covering $\tau$ and  satisfying $\tau_E^2 = c \id_E$.  In particular if $c = 1$ then $E$ is just a Real 
vector bundle in the sense of Atiyah \cite{Ati}.
\item[(b)]
If $L \to M$ is a $\ZZ_2$-equivariant line bundle, then an \emph{$L$-module}
is a vector bundle $E \to M$ with a linear map $E \to \tau^{-1}(E^*) \otimes L$ such that the composition
$$
E \longrightarrow \tau^{-1}(E^*) \otimes L \longrightarrow E  \otimes \tau^{-1}(L^*) \otimes L  
$$
is $\id_E \otimes c$, where $c \in {\mathit\Gamma}\big(\tau^{-1}(L^*) \otimes L \big) $ is the $\ZZ_2$-structure on $L$.
\item[(c)]
If $\cG=(P,Y)\To M$ is a Real bundle gerbe, then a \emph{$\cG$-module}
  is a vector bundle $E\to Y$ with a bundle map $P\otimes\pi_2^{-1}(E)\to\pi_1^{-1}(E)$ over $Y^{[2]}$ satisfying the natural associativity condition that the diagram
$$
\xymatrix{
 & P_{(y_1,y_2)}\otimes P_{(y_2,y_3)}\otimes E_{y_3} \ar[dl] \ar[dr] & \\
 P_{(y_1,y_2)}\otimes E_{y_2} \ar[dr] & & P_{(y_1,y_3)}\otimes E_{y_3} \ar[dl] \\
 & E_{y_1} &
}
$$
commutes on any triple $(y_1,y_2,y_3)\in Y^{[3]}$, together with a conjugate linear map $\tau_E \colon E \to E$ commuting with 
the Real structure on $Y$ such that the
 diagram
$$ 
\xymatrix{ 
P_{(y_1, y_2)} \otimes E_{y_2} \ \ar[d]_{\tau \otimes \tau_E} \ar[r] & \ E_{y_1} \ar[d]^-{\tau_E} \\ 
 P_{(\tau(y_1), \tau(y_2))}^* \otimes E^*_{\tau(y_2)} \ \ar[r] & \ E^*_{\tau(y_1)} } 
$$
commutes for every pair $(y_1, y_2) \in Y^{[2]}$.
\end{itemize}
\end{definition}

The case of equivariant line bundle modules is the construction of Gao and Hori 
from \cite{GaoHor}, where the orientifold isomorphism $E \to
\tau^{-1}(E^*) \otimes L$ is the action of the parity transform on the
Chan--Paton bundle $E$. The case of Real bundle gerbe modules leads to the geometric realisation $\KR_{\rm bg}(M,P)$ of twisted $\KR$-theory from~\cite{HekMurSza}.  \begin{proposition}
There is an isomorphism between the $K$-theory of  invariant $p$-gerbes $\mathcal G_p$ for $p=-1,0$, and the $K$-theory of their image 
$\partial_{p+2} \mathcal G_p$ under the connecting homomorphism.
\end{proposition}
\begin{proof}
The proof is simply a matter of translating what the respective modules mean in each case and verifying that they 
agree.  

First we consider the relationship between (a) and (b). If $L$ is the
trivial bundle on $M$, then it carries a natural $\ZZ_2$-action 
and any other action is obtained by multiplication by a Real function $c$.  Indeed this construction is 
the coboundary map $\partial_1$ in \eqref{eq:connectinghoms}.  Conversely starting with such a bundle $L$, the construction in (b) reduces to (a). 

Next we look at the relationship between equivariant line bundles and Real bundle gerbes. 
Suppose that $L \to M$ is an equivariant  line bundle. The coboundary map
$(Q,Y)=\partial_2(L)$ was constructed explicitly in
Section~\ref{sec:GaoHori} above. 
Let $F$ be a Real bundle gerbe module for $Q$. Then $\tau^{-1}(F) =
F^*$, so it has the form $F = (F_0, F_1) = ( E, \tau^{-1}(E^*))$ for a
vector bundle $E \to M$.  The Real
bundle gerbe multiplication gives isomorphisms $Q_{(i, j)}\otimes F_j \to F_i$ which are associative in the sense that the map
$$
Q_{(k, i)} \otimes ( Q_{(i, j)}\otimes F_j )  \longrightarrow Q_{(k, i)} \otimes F_i \longrightarrow F_k 
$$
is equal to 
$$
(Q_{(k, i)} \otimes  Q_{(i, j)} ) \otimes F_j  \longrightarrow Q_{(k,
  j)} \otimes F_j \longrightarrow F_k  \ .
$$
Using the definitions of $Q$ and $F$, we require   the isomorphisms
$
L \otimes \tau^{-1}(E^*) \to  E
$
or $ E \to L \otimes \tau^{-1}(E^*) $. The condition on 
$$
E \longrightarrow \tau^{-1}(E^*) \otimes L \longrightarrow E  \otimes \tau^{-1}(L^*) \otimes L
$$
follows from the associativity conditions on the bundle gerbe module, which arise
from the fact that 
$
 F_1 \simeq  Q_{(1, 0)} \otimes Q_{(0, 1)} \otimes F_1 \simeq  F_1
 $
 is the identity map and the  $\ZZ_2$-action $\tau^{-1}(L^*) = L^*$.
Furthermore, it is straightforward to check that this argument can be reversed, so that a module $E$ for $L$ defines a module $( E, \tau^{-1}(E^*))$ for the bundle gerbe $Q$. 

In both of these cases it is clear that the bijection we have
established is an isomorphism of semi-groups, and thus establishes an
isomorphism of the corresponding Grothendieck groups. 
\end{proof}

\begin{remark}
The modules on $M$ for an equivariant line bundle $(L,c)$ are the same as the Real bundle gerbe modules introduced in~\cite{HekMurSza} for the complex
trivial bundle gerbe $Q$ with non-trivial Real structure given by $(L,c)$. In other words, the Real
bundle gerbe modules for the bundle gerbe $Q$ restrict to orthogonal and quaternionic bundles on the different
connected components of the fixed point set of $(M,\tau)$, and in particular $\KR_{\rm bg}(M,Q)\simeq\KR(M;c)$.
\end{remark}

\end{document}